\documentclass[lettersize,journal]{IEEEtran}

\usepackage{amsmath,amsfonts}
\usepackage[linesnumbered,ruled,vlined]{algorithm2e}
\usepackage{array}
\usepackage{orcidlink}
\usepackage[caption=false,font=normalsize,labelfont=sf,textfont=sf]{subfig}
\usepackage{textcomp}
\usepackage{stfloats}
\usepackage{url}
\usepackage{verbatim}
\usepackage{graphicx}
\usepackage{cite}
\usepackage{booktabs}
\usepackage{multirow}
\usepackage{stfloats}
\usepackage{amsmath}
\usepackage{color}
\newtheorem{theorem}{Theorem}
\newtheorem{lemma}{Lemma}
  
\begin{document}

\title{Joint Task Partitioning and Parallel Scheduling in Device-Assisted Mobile Edge Networks}

\author{Yang Li$^{\orcidlink{0009-0007-1601-633X}}$, Xinlei Ge$^{\orcidlink{0009-0007-2546-373X}}$, Bo Lei$^{\orcidlink{0000-0002-6301-048X}}$, Xing Zhang$^{\orcidlink{0000-0003-4345-6166}}$,~\IEEEmembership{Senior Member,~IEEE,} and Wenbo Wang$^{\orcidlink{0000-0002-0911-3189}}$,~\IEEEmembership{Senior Member,~IEEE}
\thanks{Manuscript received 17 July 2023; revised 11 September 2023; accepted 1 December 2023. This work is supported by the National Science Foundation of China (NSFC) under Grant 62271062, 62071063. (Corresponding author: Xing Zhang.)}
\thanks{Y. Li, X. Ge, X. Zhang, and W. Wang are with the School of Information and Communications Engineering, Beijing University of Posts and Telecommunications, Beijing 100876, China (e-mail: ly209991@bupt.edu.cn; 2019gexinlei@bupt.edu.cn; zhangx@ieee.org; wbwang@bupt.edu.cn).}
\thanks{B. Lei is with Beijing Branch of China Telecom Co., Ltd., Beijing 100032, China (e-mail: leibo@chinatelecom.cn).}
\thanks{Copyright (c) 2023 IEEE. Personal use of this material is permitted. However, permission to use this material for any other purposes must be obtained from the IEEE by sending a request to pubs-permissions@ieee.org.}
}

\markboth{IEEE INTERNET OF THINGS JOURNAL}%
\IEEEaftertitletext{\vspace{-0.5\baselineskip}}

\maketitle

\begin{abstract}
With the development of the Internet of Things (IoT), certain IoT devices have the capability to not only accomplish their own tasks but also simultaneously assist other resource-constrained devices. Therefore, this paper considers a device-assisted mobile edge computing system that leverages auxiliary IoT devices to alleviate the computational burden on the edge computing server and enhance the overall system performance. In this study, computationally intensive tasks are decomposed into multiple partitions, and each task partition can be processed in parallel on an IoT device or the edge server. The objective of this research is to develop an efficient online algorithm that addresses the joint optimization of task partitioning and parallel scheduling under time-varying system states, posing challenges to conventional numerical optimization methods. To address these challenges, a framework called online task partitioning action and parallel scheduling policy generation (OTPPS) is proposed, which is based on deep reinforcement learning (DRL). Specifically, the framework leverages a deep neural network (DNN) to learn the optimal partitioning action for each task by mapping input states. Furthermore, it is demonstrated that the remaining parallel scheduling problem exhibits NP-hard complexity when considering a specific task partitioning action. To address this subproblem, a fair and delay-minimized task scheduling (FDMTS) algorithm is designed. Extensive evaluation results demonstrate that OTPPS achieves near-optimal average delay performance and consistently high fairness levels in various environmental states compared to other baseline schemes.
\end{abstract}

\begin{IEEEkeywords}
Device-assisted mobile edge networks, task partitioning, parallel scheduling, deep reinforcement learning (DRL), fairness.
\end{IEEEkeywords}

\section{Introduction}\label{I}

\subsection{Research Background and Motivation}\label{I-A}
\IEEEPARstart{T}{he} Internet of Things (IoT) is evolved by the 5G-enabled tactile internet \cite{1}. IoT devices often have limited computational capabilities and finite battery lives due to their small physical size and strict production cost limitations. At the same time, the recent emergence of innovative applications such as virtual reality (VR) and autonomous driving urgently require low delay and large amounts of computational resources. To ensure the quality of experience (QoE) for users, mobile edge computing (MEC) has emerged as a promising paradigm, widely acknowledged as a key solution to enhance the computation performance of IoT devices. By deploying MEC servers at the edge of radio access networks, such as cellular base stations, IoT devices can offload computationally intensive tasks to the nearby edge server (ES), thereby mitigating the energy and time costs of computation.

Extensive research has been conducted on computation offloading to enhance the computational performance of MEC networks. The research scenarios can be categorized into two types: 1) single-user scenarios \cite{2,3,4,5} and 2) multi-user scenarios \cite{6,7,8,9}. Multi-user scenarios are typically more intricate and require consideration of user cooperation or competition. The primary optimization objectives include energy consumption \cite{1,3,5}, delay \cite{10,11}, and the tradeoff between these aspects \cite{12}. Additionally, task offloading can be classified into two categories: 1) partial offloading \cite{4,5,7,8,13,14} and 2) binary offloading \cite{1,2,3,9,15,16,17,18}. Partial offloading involves users offloading a portion of their computational tasks to the nearby ES due to task separability. In the binary offloading scheme, users can transfer their entire computational tasks to the nearby ES due to task indivisibility.

Nonetheless, the aforementioned two task offloading types fail to comprehensively harness the benefits of partitionable tasks. Dividing tasks into multiple segments enables a finer scheduling approach and reduces processing delays. As shown in Fig. \ref{task}(a), a task can be decomposed into multiple task parts, which can be executed in a parallel and load-balanced way. However, if necessary, this partitioning pattern might replicate certain dependent subtasks across multiple partitions \cite{19}, potentially resulting in the repetition of processing overlapping subtasks during task execution. Consequently, striking a balance between the performance benefits derived from task deconstruction and the supplementary overhead incurred by the recurrent processing of overlapping subtasks becomes essential to optimize overall efficiency. Increasingly complex IoT applications comprise a series of subtasks, with the majority adhering to this task partitioning pattern. For instance, in VR games, inputs encompass multimodal data, yielding various outputs like visual, auditory, and haptic sensations. Another case involves monitoring vehicles in intelligent traffic systems. This task entails analyzing photos of specific vehicles regularly. The results span numerous dimensions, encompassing data like license plate particulars, driver status updates, and vehicle speed statistics. These tasks can be partitioned into parallel execution components based on the outputs. In this paper, we specialize in this task partitioning pattern.

Although MEC has advantages in task offloading, there is still a high pressure when facing massive IoT devices. As IoT technology evolves, the number of IoT devices has exploded. According to a recent report from Cisco \cite{20}, the number of devices connected to IP networks will be more than three times the global population by 2023. At the same time, MEC servers have limited resources, which means that competition for MEC server resources from users will become more intense in the future. Considering this situation, we can exploit the diversity among IoT devices, as most mobile users use less than one-third of their CPU capabilities \cite{21}. In this context, IoT devices with idle resources can be included as resource providers to provide computing resources to resource-constrained users. This is the so-called device-assisted mobile edge network (MEN), which is also the current research trend \cite{1,10,13,14,15,16,17,18,22,23,24,25,26,27,28}. The task partitioning pattern studied in this paper is well-suited for device-assisted MENs. Partitioning a task into multiple task parts for parallel processing on multiple IoT devices with idle computational resources can ease the computational pressure on ESs while improving the computational resource utilization of IoT devices and the QoE for users.

\begin{figure*}[!t]
\centering
\subfloat[]{\includegraphics[width=0.5\textwidth]{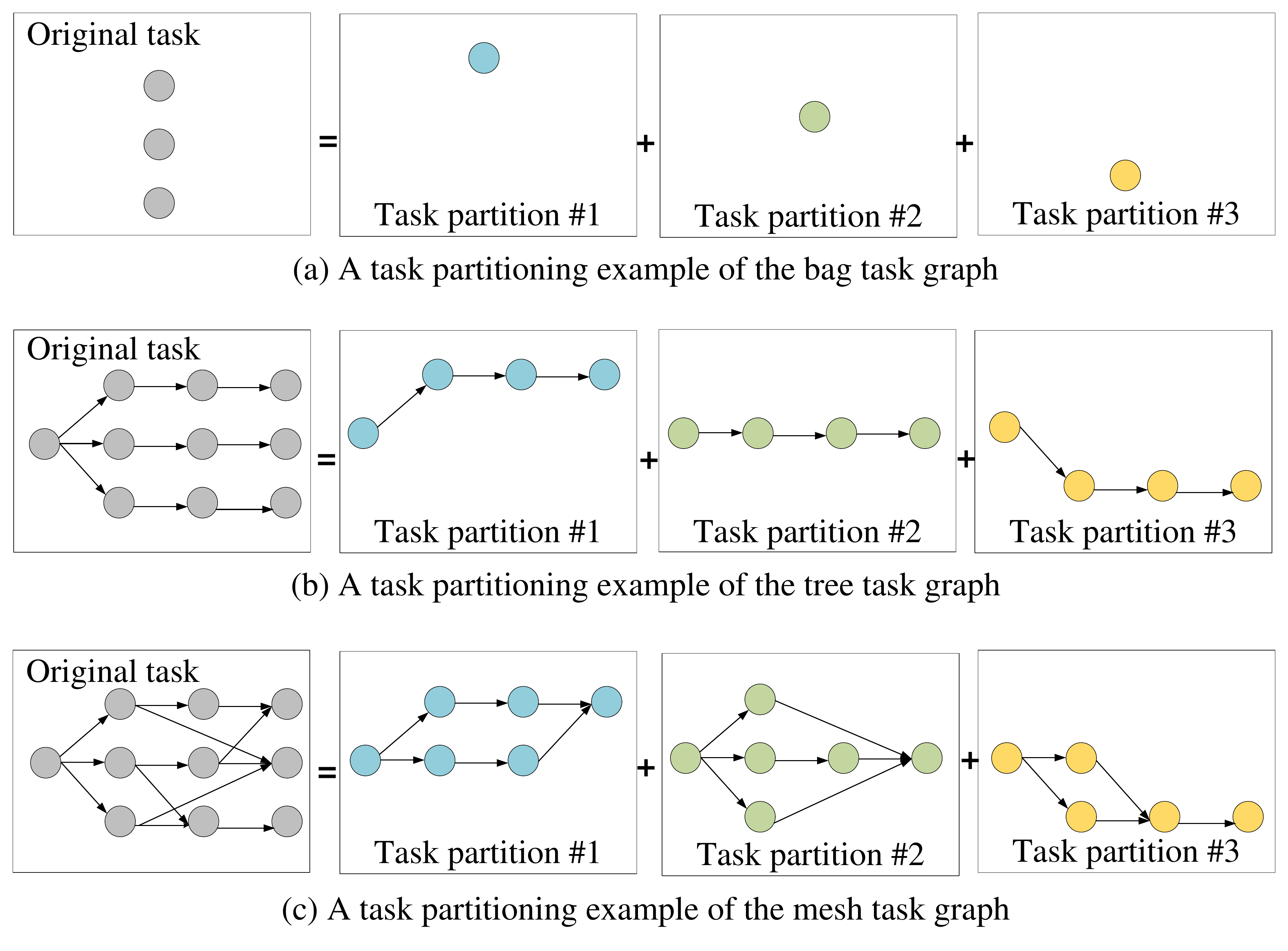}%
\label{task1}}
\subfloat[]{\includegraphics[width=0.5\textwidth]{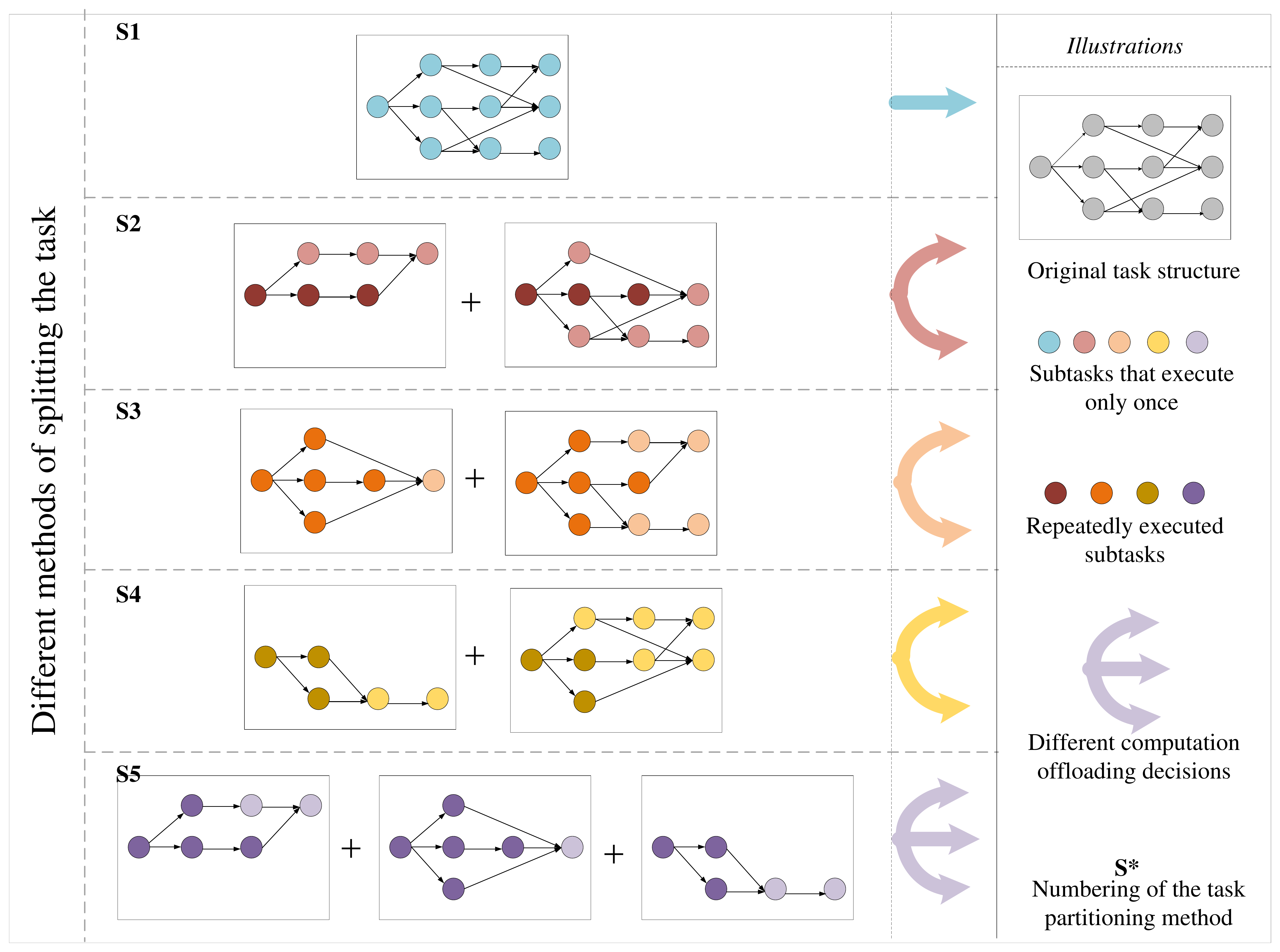}%
\label{task2}}
\caption{The task model studied in this paper. (a) An illustration of three typical task types and corresponding task partitioning examples in the system model. (b) An illustration of five ways to split a mesh type task.}
\vspace{-0.5cm}
\label{task}
\end{figure*}

As shown in Fig. \ref{problem}, employing the task partitioning model explored in this study within the device-assisted MEN system enhances user experiences, including reduced delay, and optimizes resource utilization of end-side devices. Nevertheless, novel technical challenges emerge. Firstly, potential unfairness among users may arise due to the varying availability of idle resources in auxiliary IoT devices. This implies that the selection of task-scheduling strategies must encompass additional considerations. For instance, when selecting computing nodes for task partitions of users, minimizing overall task processing delay and upholding fairness among users is imperative. This precludes a sole emphasis on maximizing system-wide utility at the detriment of individual user device performance. On the other hand, the actions of task partitioning have an impact on subsequent parallel scheduling strategies. As shown in Fig. \ref{task}(b), multiple task partitioning approaches exist. Should way S1 be chosen, a lone computing node suffices for the task; in contrast, opting for ways S2-S4 requires two compute nodes, and selecting way S5 mandates three compute nodes. Given that the overall completion time of the task is determined by the last partition to finish, variations arise in the selection of parallel scheduling strategies for each task partitioning approach within S2-S4. In multi-user device-assisted MEN, joint optimization of task partitioning and parallel scheduling strategies while maintaining fairness among users is a challenging but rewarding problem.

\begin{figure*}[t]
\centerline{\includegraphics[width=0.8\textwidth]{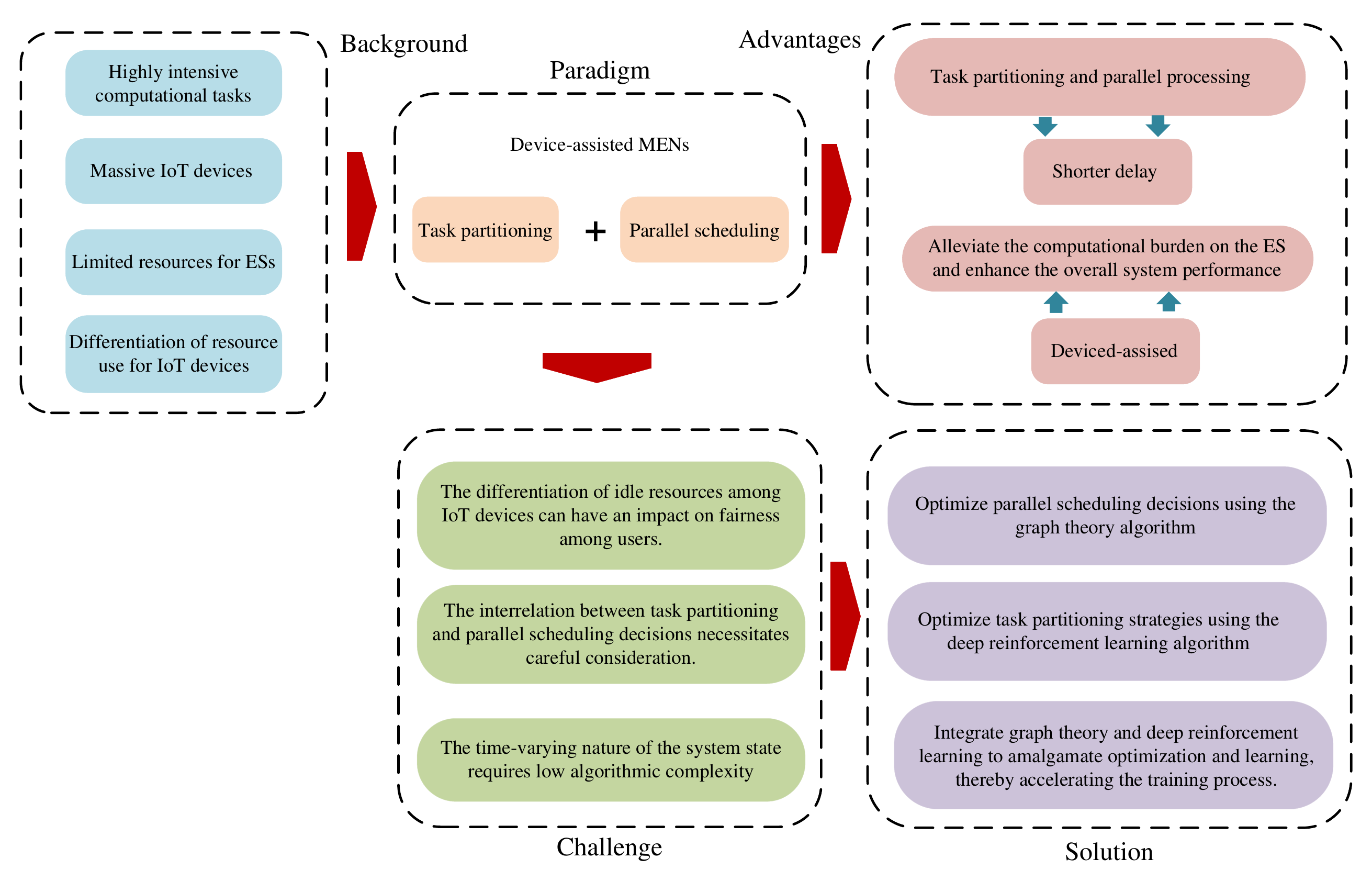}}
\caption{Roadmap of the challenge, problem, and solution for our study.}
\vspace{-0.5cm}
\label{problem}
\end{figure*}

\subsection{Related Work}\label{I-B}
\textit{1) Optimization Problems in MEC:} Currently, significant research efforts are centered on this focal area. Task offloading can be categorized into two forms: 1) partial offloading \cite{4,5,7,8,13,14} and 2) binary offloading \cite{1,2,3,9,15,16,17,18}. For instance, aiming to optimize delay and energy consumption, Wang \textit{et al.} \cite{4} derived the optimal data segmentation strategy within partial offloading mode. The authors in \cite{29} considered optimizing video-based AI inference tasks in a multi-user MEC system. They determined whether the DNN inference task should be executed entirely on the device or MEC server and optimized resource allocation to achieve diminished inference delays, lowered energy consumption, and enhanced recognition accuracy. Nevertheless, the above-mentioned two task offloading approaches inadequately leverage the advantages of partitionable tasks. Partitioning tasks into multiple parts enables a more precise scheduling methodology and mitigates processing delays. For example, in \cite{12}, the authors proposed a framework that splits a task into subchunks and offloads them to a common ES using multiple radio access technologies. Similarly, \cite{30} investigated the problem of partitioning multiple tasks into subtasks and scheduling them in parallel. The authors employed the Advantages Actor-Critic (A2C) method to select computation nodes for each subtask and demonstrated an optimal task split ratio function. However, the entirety of the aforementioned research presumed that tasks could be fractionated into numerous parallel processing elements with arbitrary ratios. In reality, the multiple subtasks constituting a task possess interdependencies and cannot be divided indiscriminately. In this paper, we propose partitioning the task into multiple segments in parallel, guided by the outputs, as shown in Fig. \ref{task}. Subsequently, we allocate each task partition to distinct computational nodes for processing. As such, joint optimization of task partitioning actions and parallel scheduling strategies becomes imperative. Additionally, the optimization aim primarily revolves around energy consumption \cite{1,3,5}, delay \cite{8,10,11}, and the tradeoff between these two aspects \cite{12}. For example, constrained by energy consumption, The authors in \cite{8} introduced a unified resource allocation strategy for communication and computation to mitigate delay. Nevertheless, the preceding studies mentioned earlier, primarily concentrating on enhancing system-level performance, might result in unfairness among users \cite{31}. Conversely, this paper enhances system performance while preserving fairness among users.

\textit{2) Device-Assisted MENs:} In the previous work, users only had the option to offload computation tasks to the edge cloud. For example, in reference \cite{32}, researchers explored two offloading approaches: 1) local computing and 2) edge offloading, enabling users to offload computation tasks to the edge cloud via wireless channels. However, the idle computing resources of IoT devices are ignored. In other words, users who need computation offloading can utilize the idle computational resources of other IoT devices and offload some of the tasks to IoT devices with idle computational resources.

Subsequently, inspired by the fact that idle IoT devices can be used as computational nodes to provide computational services to users, device-assisted MEC emerges as a promising strategy to enhance offloading benefits. Consequently, users have three options for task processing locations: local, the edge server, and IoT devices with idle resources. For instance, reference \cite{33} explored cooperative computing in a non-orthogonal multiple access (NOMA)-enabled MEC system comprising an edge server, a user equipment (UE), and multiple full-duplex helpers. An adaptive cooperative computing strategy was devised. The authors in \cite{34} investigated the joint task assignment, communications rate, as well as computation frequency allocation within a D2D-enabled multi-helper MEC system under the assumption of binary task offloading. An efficient algorithm based on convex relaxation was proposed to minimize the cumulative computational delay. However, both \cite{33} and \cite{34} exclusively addressed the single-user multi-helper MEC model, and generalizing the outcomes to multi-user scenarios could be challenging.

Indeed, in multi-user scenarios, improving the overall system performance while ensuring fairness among users is challenging, attributed to the competitive dynamics among users. Additionally, due to the time-varying nature of idle resources and channel quality of IoT devices, frequent algorithm execution requires the proposed scheme to have low complexity. Consequently, joint optimization of task partitioning actions and parallel scheduling strategies for multiple offloaded users with competing relationships in the device-assisted MEN has not been fully investigated.

\textit{3) DRL in MEC:} DRL has emerged as a promising avenue for tackling real-time computation offloading challenges within MEC networks. Current DRL-based approaches adopt a value-based or policy-based strategy to learn the optimal translation from the “state” (e.g., time-varying system parameters) to the “action” (e.g., offloading decisions and resource allocation) \cite{40}. Widely employed value-based DRL techniques include deep Q-learning network (DQN) \cite{35}, double DQN \cite{36}, and dueling DQN \cite{37}, where a DNN is trained to estimate the state-action value function. Nevertheless, DQN-based methodologies become costly when confronted with sizable feasible discrete offloading actions, such as exponential growth concerning the count of IoT devices. In response, recent investigations have embraced a policy-based methodology, including the actor-critic DRL \cite{30} and the deep deterministic policy gradient (DDPG) techniques \cite{38}, \cite{39}. These methods directly formulate the optimal mapping from input state to output action using a DNN. For instance, in \cite{38}, an IoT device exclusively undertakes discrete offloading actions, encompassing integer offloading decisions and disregarded transmit power and offloading rate. Subsequently, an actor-critic DRL approach is employed to acquire the optimal mapping from continuous input states to discrete output actions. \cite{39} established two distinct learning components sequentially, aimed at producing discrete offloading decisions and continuous resource allocation. Precisely, \cite{39} utilized an actor DNN to yield the resource allocation solution, coupled with a DQN-based critic network to opt for the discrete offloading action. Analogous to \cite{35,36,37}, gauging the state-action value function within the critic network becomes arduous when faced with a multitude of potential offloading actions. Additionally, all the previously mentioned studies concentrated on binary offloading and partial offloading. Therefore, these concepts cannot be immediately applied to the task partitioning pattern investigated within this paper. The OTPPS framework proposed in this paper uses an actor DNN to generate a small number of candidate task partitioning actions. Subsequently, it employs a critic module, founded upon an enhanced graph theory algorithm, to discern the optimal task partitioning action via analytical resolution of the corresponding parallel scheduling problems. Thanks to the accurate evaluation of actions by the critic module, OTPPS is able to quickly converge to the optimal solution, even when the actor DNN provides only very few actions for the critic to select.

\subsection{Contribution and Organization}\label{I-C}
This study delves into a multi-user multi-helper MEC network as depicted in Fig. \ref{system}, where users’ tasks adhere to the structure illustrated in Fig. \ref{task}. The objective is to devise an algorithm for real-time task partitioning and parallel scheduling, ensuring fairness among users while minimizing overall task processing latency. To tackle this challenge, we introduce an innovative framework termed online task partitioning and parallel scheduling (OTPPS), synergizing the strengths of graph theory algorithms and DRL. OTPPS empowers real-time online optimization decision-making amidst dynamic system conditions. The primary contributions of this research encompass:
\begin{enumerate}
\item{A fairness-aware delay minimization problem is proposed to minimize the maximum normalized delay among all UEs, distinguishing this research from previous works.}
\item{To solve this problem, this paper proposes a novel OTPPS framework that combines the advantages of graph theory algorithms and DRL. Specifically, we integrate graph theory-based optimization and DRL methodologies to tackle the non-linear programming (NLP) issue for each time frame with very low computational complexity.}
\item{OTPPS employs an actor-critic structure to address the per-frame NLP problem. The actor module is a DNN that learns the optimal task partitioning action based on the input environment parameters. Meanwhile, the critic module evaluates the task partitioning action by resolving the parallel scheduling issue with an enhanced graph theory algorithm. In contrast to the traditional actor-critic structure that incorporates a model-free DNN in the critic module, the proposed framework capitalizes on model-derived insights for precise action evaluation, thus reaping the benefits of heightened robustness and swifter convergence in the DRL training procedure.}
\item{OTPPS employs a sliding threshold quantization approach to yield candidate task partitioning actions, effectively striking a balance between exploration and exploitation facets (i.e., performance-focused or diversity-oriented) in the DRL algorithm design to expedite training convergence. Concurrently, simulation results corroborate the capacity of the suggested mechanism to enhance system-wide performance while upholding user fairness, and verify the convergence and efficacy of the proposed algorithm.}
\end{enumerate}

The remainder of this paper is organized as follows. Section II introduces the system model and problem formulation. Sections III and IV provide the OTPPS framework and describe the underlying algorithms in detail. Section V presents the simulation results and performance analysis. Finally, Section VI concludes this paper.

\begin{figure*}[b]
\centerline{\includegraphics[width=0.7\textwidth]{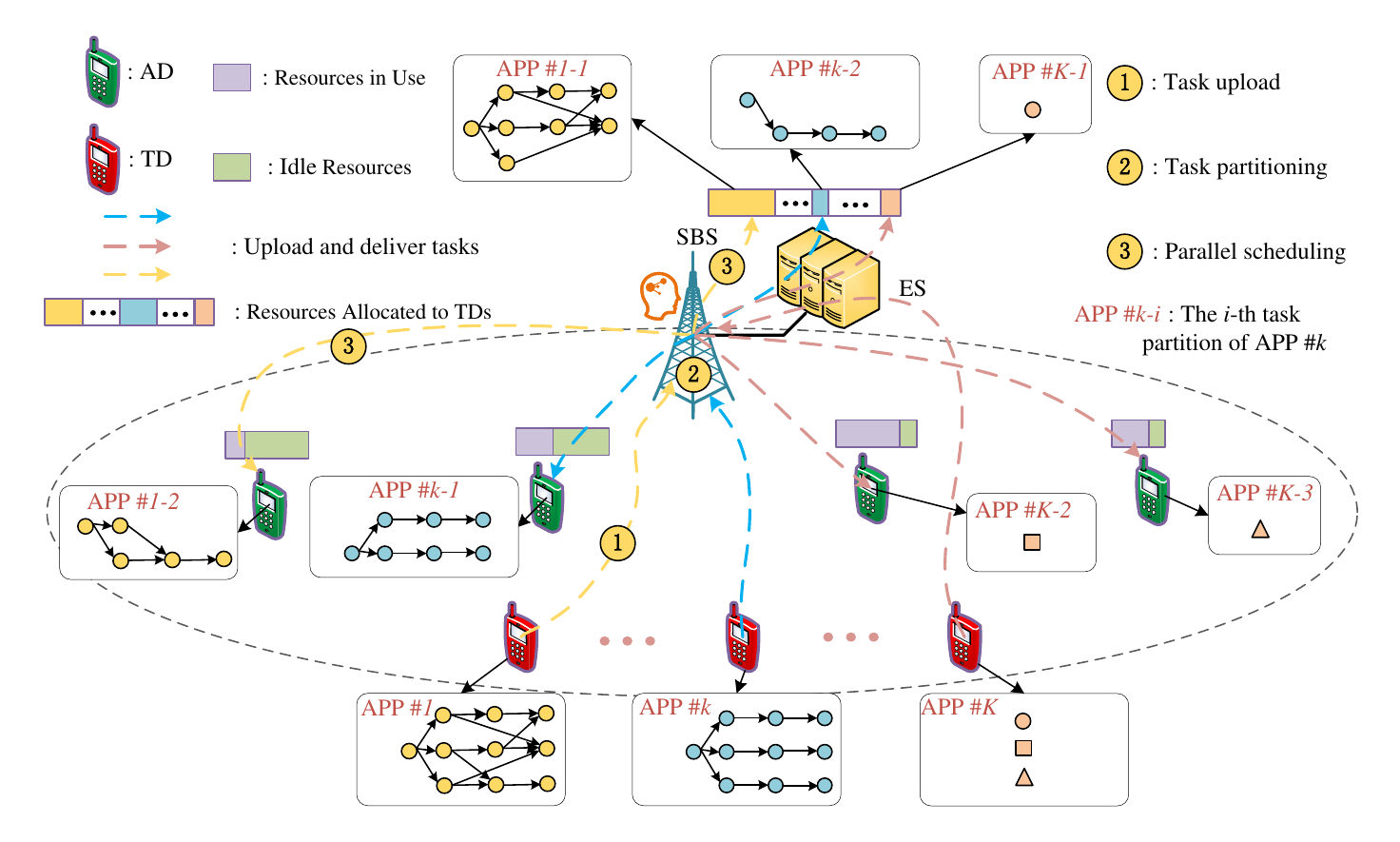}}
\caption{An illustration of task partitioning and parallel scheduling in the device-assisted MEN.}
\label{system}
\end{figure*}

\begin{table}[]
\centering
\caption{Summary of system notations}
\label{notation}
\resizebox{\columnwidth}{!}{%
\begin{tabular}{@{}ll@{}}
\toprule
Notation                                             & Description                                                                                                                  \\ \midrule
$\mathcal{N}$                                        & The set of task IoT devices                                                                                                  \\
$\mathcal{M}$                                        & The set of auxiliary IoT devices                                                                                             \\
$\mathcal{S}$                                        & \begin{tabular}[c]{@{}l@{}}The set of resource devices (including ES 0 and auxiliary IoT\\ devices)\end{tabular}             \\
$\mathcal{T}_n$                                      & The task of TD $n$                                                                                                           \\
$N_{n,res}$                                          & The number of results for $\mathcal{T}_n$                                                                                    \\
$r_{n,i}$                                            & The $i$th result of $\mathcal{T}_n$                                                                                          \\
$\boldsymbol{\mathcal{T}}_n^t$                       & The set of task partitions for TD $n$ in the $t$th time frame                                                                \\
$\mathcal{T}_{n,i}^t$                                & The $i$th task partition of TD $n$ in the $t$th time frame                                                                   \\
$N_{n,part}^t$                                       & The number of task partitions for TD $n$ in the $t$th time frame                                                             \\
$\boldsymbol{\mathcal{U}}_t$                         & The task partitioning strategy in the $t$th time frame                                                                       \\
$u_t^{n,i}$                                          & An indicator variable denotes that $r_{n,i}$ is output by $\mathcal{T}_{n,u_t^{n,i}}^t$                                      \\
$\boldsymbol{\mathcal{W}}_t$                         & The parallel scheduling strategy in the $t$th time frame                                                                     \\
$w_t^{n,i}$                                          & An indicator variable denoting the processing position of $\mathcal{T}_{n,i}^t$                                              \\
$\boldsymbol{\mathcal{H}}^t$                         & \begin{tabular}[c]{@{}l@{}}The set of wireless channel fading coefficients\end{tabular}            \\
$\boldsymbol{\mathcal{F}}_0^t$                       & \begin{tabular}[c]{@{}l@{}}The set of computation resources allocated by ES 0 to TDs\end{tabular} \\
$\boldsymbol{\mathcal{F}}_{\mathcal{M}}^t$           & The set of idle computational resources for ADs                                                                              \\
$L_n^t$                                              & The completion time of $\mathcal{T}_n$ in the $t$th time frame                                                               \\
$L_{n,i}^t$                                          & The completion time of $\mathcal{T}_{n,i}^t$                                                                                 \\
$T_n^{max}$                                          & The maximum acceptable delay for $\mathcal{T}_n$                                                                             \\
$\eta_n^t$                                           & The normalized delay for $\mathcal{T}_n$ in the $t$th time frame                                                             \\
$\hat{\boldsymbol{\mathcal{U}}}_{t}$                 & Relaxed task partitioning action                                                                                             \\
$\widetilde{\boldsymbol{\mathcal{U}}}_t^*$           & The normalized value of $\boldsymbol{\mathcal{U}}_t^*$                                                                       \\
$\overline{\boldsymbol{\mathcal{U}}}_t^{q^{\prime}}$ & The $q^{\prime}$th candidate task partitioning action                                                                        \\
$\overline{\eta}_t^{q^{\prime}}$                     & The normalized delay achieved with $\overline{\boldsymbol{\mathcal{U}}}_t^{q^{\prime}}$                                      \\
$Q$                                                  & The number of quantized task partitioning actions                                                                            \\
$Q'$                                                 & The number of candidate task partitioning actions                                                                            \\
$\theta_t$                                           & The parameters of the DNN                                                                                                    \\
$\Omega$                                             & The training interval of the DNN                                                                                             \\ \bottomrule
\end{tabular}%
}
\end{table}

\section{System model and problem statement}
\subsection{Network Model}\label{II-A}
Consider an IoT system comprising a smart base station (SBS) integrated with an edge server (ES 0) and multiple IoT devices. Similar to the previous work \cite{1}, the IoT devices in the network are categorized into two types: task IoT devices (TDs) and auxiliary IoT devices (ADs). The set $\mathcal{N}=\{1, 2, ..., N\}$ of TDs consists of all IoT devices which have limited computation resources and computation-intensive tasks to perform, so TDs can only offload their tasks. On the other hand, ADs consist of a set of auxiliary IoT devices denoted as $\mathcal{M}=\{1, 2, ..., M\}$, which possess sufficient computation capabilities. Each IoT device in $\mathcal{M}$ has idle computation resources to assist a TD in performing its task. For simplicity later, we refer to ES 0 and ADs together as resource devices (RDs), denoted as $\mathcal{S}=\{0\}\cup \mathcal{M}$.

As shown in Fig. \ref{task}(a), we assume that the tasks of TDs require to output mutipule results and can be partitioned \cite{19}. Thus, the task of TD $n$ can be partitioned into multiple task parts, each of which is processed on an IoT device or the edge server. It is important to note that each IoT device in $\mathcal{M}$ can serve at most one TD within a specific time frame \cite{1,52}. Consequently, a task of TD $n$ can be processed collaboratively by multiple IoT devices and ES 0.

A simple illustration of the considered scenario is shown in Fig. \ref{system}. It illustrates the process of task uploading, task partitioning, and parallel scheduling, which alleviates the workload of the ES 0 by utilizing idle computing resources of IoT devices and reduces delay through parallel processing. The important notations used in the rest of this paper are summarized in Table \ref{notation}.

\subsection{Task Model}\label{II-B}
This study explores a general partition pattern, as shown in Fig. \ref{task}(a), which duplicates dependent subtasks to multiple partitions when necessary \cite{19}. This partition pattern allows for various methods of task decomposition, as illustrated in Fig. \ref{task}(b).

The task of TD $n$ is denoted as $\mathcal{T}_n = (Z_n, C_n)$, where $Z_n$ represents the data size offloaded over wireless links, and $C_n$ is the number of CPU cycles required to process the entire task. Specifically, we consider result-partitioned-oriented applications, requiring to output multiple results that are independent of each other but dependent on subtasks that may overlap. Thus, these applications can be arbitrarily partitioned for parallel processing based on the output. $N_{n, res}$ represents the number of results for task $\mathcal{T}_n$ and $N_{n, part}^t$ denotes the number of task partitions after $\mathcal{T}_n $ has been partitioned in the $t$th time frame, so $N_{n, part}^t \le N_{n, res}$. $\mathcal{R}_n=\{r_{n,1},r_{n,2},...,r_{n,N_{n, res}}\}$ denotes the set of results for $\mathcal{T}_n$. $\boldsymbol{\mathcal{T}}_n^t=\{\mathcal{T}_{n,1}^t,\mathcal{T}_{n,2}^t,...,\mathcal{T}_{n, N_{n, part}}^t\}$ denotes the set of task partitions for TD $n$ in the $t$th time frame, where $\mathcal{T}_{n,k}^t=(Z_{n,k}^t, C_{n,k}^t)$ denotes the $k$th task partition. We denote $Z_{n,k}^t$ and $C_{n,k}^t$ as the data size offloaded over wireless links and the required CPU cycles of $\mathcal{T}_{n,k}^t $, respectively.

The identity of $\mathcal{T}_{n,k}^t$ is denoted as ${\bf{q}}_{n,k}^t = [q_{n,k,1}^t,q_{n,k,2}^t,...,\\q_{n,k,{N_{n,res}}}^t]$, where $q_{n,k,i}^t \in \{ 0,1\}$ is an indicator variable. If $q_{n,k,i}^t = 1$, it indicates that $\mathcal{T}_{n,k}^t$ needs to output the result $r_{n,i}$; otherwise, it is zero. Since each task partition needs to output at least one result, every task partition must satisfy the condition: $\sum\limits_{i = 1}^{{N_{n, res}}} {q_{n,k, i}^t} > 0$. Due to the dependencies among subtasks, when $\mathcal{T}_{n,{k_1}}^t,{\cal T}_{n,{k_2}}^t$ and 
$\mathcal{T}_{n,{k_3}}^t$ satisfy ${\bf{q}}_{n,{k_1}}^t = {\bf{q}}_{n,{k_2}}^t{\rm { + }}{\bf{q}}_{n,{k_3}}^t$, it can be deduced that $Z_{n,{k_1}}^t \le Z_{n,{k_2}}^t + Z_{n,{k_3}}^t$ and $C_{n,{k_1}}^t \le C_{n,{k_2}}^t + C_{n,{k_3}}^t$, as illustrated in Fig. \ref{task}(b) (S1 and S2).

\subsection{Communication Model}\label{II-C}
Based on the aforementioned descriptions and definitions, the communication process for task partitioning and parallel scheduling in the device-assisted MEN consists of the following steps (as illustrated in Fig. \ref{commPro}).

\begin{figure}[t]
\centerline{\includegraphics[width=0.5\textwidth]{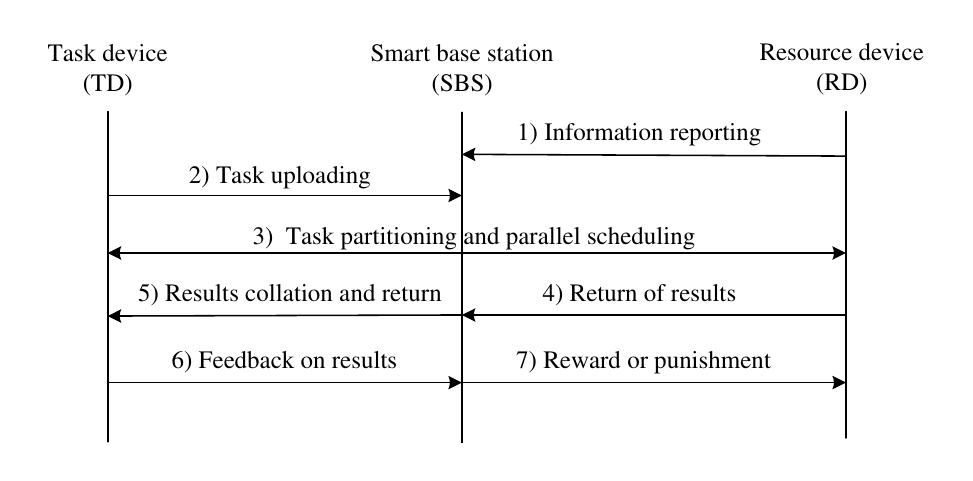}}
\caption{An illustration of the communication process.}
\label{commPro}
\end{figure}

\begin{enumerate}
\item{At the end of each time frame, all ADs provide a report to the SBS regarding the available computational resources and their respective positions for the next time frame. Additionally, ES 0 reports the allocation of its resources. Simultaneously, the SBS measures the channel fading coefficients between itself and all ADs based on the received signals.}
\item{At the beginning of each new time frame, the TDs generate tasks and offload them to the SBS. They also provide information about their available computational resources and current positions. Similarly, based on these received signals, the SBS measures the channel fading coefficients between itself and all TDs.}
\item{The intelligent agent of SBS makes task partitioning and parallel scheduling decisions utilizing the collected information and then schedules the divided task partitions to different computing nodes (including TDs and RDs) for parallel processing. It is assumed that only TD $n$'s task can be processed on TD $n$.}
\item{RDs complete the calculation and return the results to the SBS.}
\item{The SBS transmits the received results back to the corresponding TDs.}
\item{Each TD provides a report to the SBS indicating that the success of computation offloading, specifically whether each sub-result meets the expected criteria.}
\item{If the offloading task partition is successfully accomplished, the SBS will provide a predetermined reward to the RD. However, in the event of failure, the RD would be subjected to paying a penalty to the SBS.}
\end{enumerate}

This paper primarily focuses on steps 1 to 5 of the proposed process. The work in \cite{41} addressed the design of rewards and penalties for steps 6 and 7. The overall process primarily revolves around OFDMA-based wireless communication.

The wireless transmission process of task data involves the uplinks between TDs and the SBS, the downlinks between the SBS and ADs, and the downlinks between the SBS and TDs. The network spectrum of the entire system is divided into multiple channels denoted as $\mathcal{J}=\{1,2,..., J\}$, which are managed by the SBS. Assuming that the downlink and uplink transmissions experience the same noise, the maximum achievable rate (in bps) for the uplink and downlink over an additive white Gaussian noise (AWGN) channel can be obtained as follows:
\begin{equation}
\label{eq1}
{R^{ul}}(n,t) = W{\log _2}(1 + \frac{{{p_n}|{h_{ul}}(n,t){|^2}}}{{d_{n,0}^{{\beta _l}}{N_0}}}),
\end{equation}
\begin{equation}
\label{eq2}
{R^{dl}}(n,t) = W{\log _2}(1 + \frac{{{p_0}|{h_{dl}}(n,t){|^2}}}{{d_{n,0}^{{\beta _l}}{N_0}}}),
\end{equation}
and
\begin{equation}
\label{eq3}
{R^{dl}}(m,t) = W{\log _2}(1 + \frac{{{p_0}|{h_{dl}}(m,t){|^2}}}{{d_{m,0}^{{\beta _l}}{N_0}}}).
\end{equation}
The three equations presented above represent the maximum transmission rate of the uplink of TD $n$, the downlink of TD $n$, and the downlink of AD $m$ in the $t$th time frame. In these equations, $W$ denotes the channel bandwidth, $N_0$ represents the noise power, $\beta_l$ is the path loss exponent, and $h_{ul}(n,t), h_{dl}(n,t),$ and $h_{dl}(m,t)$ represent the channel fading coefficients of the uplink of TD $n$, the downlink of TD $n$, and the downlink of AD $m$ in the $t$th time frame, respectively. The variables $d_{n,0}$ and $d_{m,0}$ denote the distances between TD $n$ and the SBS, and AD $m$ and the SBS, respectively. Additionally, $p_n$ and $p_0$ represent the transmitting power of TD $n$ and the SBS, respectively. It is assumed that the IoT system is quasi-static, meaning that the system state remains unchanged during each time frame.

\subsection{Computing Model}\label{II-D}
Considering that computation offloading and resource allocation decisions have different time sensitivities, simultaneously making these decisions will lead to the offloading process waiting for resource allocation, thus reducing its real-time requirement. Alternatively, it will cause the system to orchestrate resources too frequently, resulting in high operating expenses and instability \cite{19}. Therefore, resource allocation and computation offloading decisions generally operate on different timescales. At the time of the offloading decision, the resource allocation is typically already determined. In this paper, our focus is on the offloading decision. We denote ${\boldsymbol{{\cal F}}_0^t} = \{ f_{0,1}^t,f_{0,2}^t,...,f_{0, N}^t\} $ as the computation resources allocated by the ES 0 to TDs in the $t$th time frame. ${\boldsymbol{{\cal F}}_{\mathcal{M}}^t} =\{f_m^t| m\in \mathcal{M}, t=1,2,\cdots\}$ represents the idle computational resources that AD $m$ can provide to the TDs in the $t$th time frame. ${\boldsymbol{{\cal F}}_\mathcal{N}} =\{f_n| n\in \mathcal{N}\}$ represents the computing resources of TD $n$.

To indicate the processing location of a task partition, we introduce an indicator variable $w^{n, i}_t$, which denotes the processing location of the $i$th task partition of TD $n$ at the $t$th time frame. Parallel scheduling decision is defined as ${\boldsymbol{{\cal W}}_t} = \{ \textbf{w}^1_t,\textbf{w}^2_t,..., \textbf{w}^N_t\} $, where $\textbf{w}^n_t = \{ w^{n,i}_t \in {\cal S}\cup \{M+1\}|i \in \{ 1,2,...,N_{n,part}^t\} \} $. $ w^{n, i}_t = M+1$ indicates that $\mathcal{T}_{n,i}^t$ is executed on TD $n$. Therefore, the computational delay $L_{n,i}^{t,comp}(w^{n,i}_t)$ of the task partition ${\cal T}_{n,i}^t$ can be expressed as
\begin{equation}
\label{eq4}
L_{n,i}^{t,comp}(w^{n,i}_t) = \left\{ \begin{array}{l}
\frac{{C_{n,i}^t}}{{f_{0,n}^t}},\quad {\rm{if}}\;w^{n,i}_t = 0,\\
\frac{{C_{n,i}^t}}{{f_m^t}}, \quad {\rm{if}}\;w^{n,i}_t = m,m \in {\cal M},\\
\frac{{C_{n,i}^t}}{{{f_n}}},\quad {\rm{if\ }}w^{n,i}_t = M+1.
\end{array} \right.
\end{equation}

\subsection{Delay Analysis}\label{II-E}
Section \ref{II-C} discusses the various stages involved in IoT task processing, including uploading input data, task partitioning, parallel scheduling, and execution, as well as results return. The transfer time for results return is considered negligible, as the result after processing is typically small \cite{31,42}. Furthermore, we do not account for the time required for algorithm execution and task partitioning. In Section \ref{V-B}, it is illustrated that the algorithm execution time is negligible. Subsequently, we provide a detailed description of the main components contributing to the overall delay.
\begin{enumerate}
\item{Uploading input data: The delay in the $t$th time frame for TD $n$ to upload IoT task data is defined as $L_n^{t,ul} = \frac{{{Z_n}}}{{{R^{ul}}(n,t)}}$.}
\item{Parallel scheduling and execution: Assume that in the $t$th time frame, the intelligent agent partitions the task $\mathcal{T}_n$ into $N_{n, part}^t$ task parts, denoted as $\{ {\cal T}_{n, i}^t|i \in \{ 1,2,..., N_{n, part}^t\} \}$, and schedules these task partitions to different locations for parallel processing. The transmission delay $ L_{n,i}^{t,trans}(w^{n,i}_t)$ of ${\cal T}_{n,i}^t$ can be expressed as
\begin{equation}
\label{eq5}
L_{n,i}^{t,trans}(w^{n,i}_t) = \left\{ \begin{array}{l}
0,\;\;\; \qquad{\rm{if\;}}w^{n,i}_t = 0,\\
\frac{{Z_{n,i}^t}}{{{R^{dl}}(m,t)}},\; {\rm{if\;}}w^{n,i}_t = m,m \in {\cal M}\\
\frac{{Z_{n,i}^t}}{{{R^{dl}}(n,t)}}, \;\; {\rm{if\;}}w^{n,i}_t = M+1.
\end{array} \right.
\end{equation}
The computation time of ${\cal T}_{n, i}^t$ can be obtained using Equation (\ref{eq4}). Therefore, the completion time of ${\cal T}_{n, i}^t$ can be calculated as
\begin{equation}
\label{eq6}
L_{n,i}^t = L_n^{t,ul} + L_{n,i}^{t,trans} + L_{n,i}^{t,comp}.
\end{equation}
}
\end{enumerate}

The completion of task ${{\cal T}_n}$ depends on receiving results from all its task partitions, and the overall completion time of the task is determined by the last partition to finish. Therefore, the completion time of task ${{\cal T}_n}$ can be calculated as
\begin{equation}
\label{eq7}
L_n^t = \max \{ L_{n,1}^t,L_{n,2}^t,\cdots,L_{n,N_{n,part}^t}^t\}.
\end{equation}

\subsection{Problem Formulation}\label{II-F}
In the considered scenario, we assume that only the channel fading coefficients ${\boldsymbol{{\cal H}}^t} = \{ {h_{ul}}(n,t),{h_{dl}}(n,t),{h_{dl}}(m,t)|n \in {\cal N},m \in {\cal M}\} $ and the computation resources $\boldsymbol{{\cal F}}_0^t$ and $\boldsymbol{{\cal F}}_{\cal M}^t$ provided by RDs are time-varying, while the others are fixed parameters. We formulate a min-max normalized delay problem to achieve fair normalized delay (ND) among TDs while reducing the overall task completion time. The ND of TD $n$ is defined as $\eta _n^t = (\frac{{L_n^t}}{{T_n^{max}}}>1)?\infty:\frac{{L_n^t}}{{T_n^{max}}}$, where $T_n^{max}$ represents the maximum acceptable delay for task ${{\cal T}_n}$. If $L_n^t$ exceeds $T_n^{max}$, $\eta _n^t$ is set to infinity, ensuring timely task completion. This is crucial because many tasks are highly sensitive to latency, and processing tasks beyond their deadlines becomes meaningless for users \cite{56,57,58,59}. Specifically, we design a partitioning strategy ${\boldsymbol{{\cal U}}_t} = \{ {\bf{u}}^n_t = [u^{n,1}_t,u^{n,2}_t,...,u^{n,{N_{n, res}}}_t]|n \in {\cal N}\} $ for partitioning TDs' tasks in the $t$th time frame and a parallel scheduling strategy ${\boldsymbol{{\cal W}}_t} = \{ {\bf{w}}^n_t = [w^{n,1}_t,w^{n,2}_t,...,w^{n, N_{n, part}^t}_t]|n \in {\cal N}\} $ for task parts after partitioning to efficiently utilize IoT system resources. ${\bf{u}}^n_t$ represents the task partitioning strategy of TD $n$, where $u^{n,i}_t \in \{ 1,2,...,N_{n,part}^t\} $ represents that the $i$th result of task ${{\cal T}_n}$ will be output by its $u^{n,i}_t$th task partition. We assume that the TDs' tasks remain unchanged over time, and the execution programs for these tasks are pre-stored on the ES, which is a common practice in IoT systems such as smart factories, industrial parks and wireless sensor networks. Thus the optimization problem can be formulated as follows:
\begin{equation}
\begin{aligned}\label{P}
\mathcal{P}_1:\mathop {\emph{min} }\limits_{{\boldsymbol{{\cal U}}_t},{\boldsymbol{{\cal W}}_t}}& \;\mathop {\emph{max} }\limits_{n \in {\cal N}} \;\eta _n^t({\boldsymbol{{\cal H}}^t},\boldsymbol{{\cal F}}_0^t,{\boldsymbol{{\cal F}}_{{{\cal M}}}^t},{\boldsymbol{{\cal U}}_t},{\boldsymbol{{\cal W}}_t}) \\
\mbox{s.t.}\;&(a):
{\boldsymbol{{\cal U}}_t} = \{ {\bf{u}}_t^n|n \in {\cal N}\},  \\
&(b):{\bf{u}}_t^n = [u^{n,1}_t,u^{n,2}_t,...,{u^{n,N_{n,res}}_t}],\forall n \in {\cal N}, \\
&(c):u^{n,i}_t \in \{ 1,2,...,N_{n,part}^t\} ,\;\forall n \in {\cal N},\\
&(d):{\sum\limits_{i = 1}^{{N_{n,res}}} {{{\bf{1}}_{u^{n,i}_t = j}} \ge 1} } ,\forall j \in \{ 1,...,N_{n,part}^t\},\\
&(e):N_{n,part}^t \le {N_{n,res}},\\
&(f):{\boldsymbol{{\cal W}}_t} = \{ {\bf{w}}^n_t|n \in {\cal N}\},\\
&(g):{\bf{w}}^n_t = [w^{n,1}_t,...,w^{n,N_{n,part}^t}_t],\;\forall n \in {\cal N},\\
&(h):w^{n,i}_t \in {\cal S}\cup \{M+1\}\;\;\forall n \in {\cal N},\forall i,\\
&(i):\sum\limits_{n \in {\cal N}} {\sum\limits_{i = 1}^{N_{n,part}^t} {{{\bf{1}}_{w^{n,i}_t = m}} \in \{ 0,1\} } } ,\;\forall m \in {\cal M},
\end{aligned}
\end{equation}
where ${{\bf{1}}_{\{  \cdot \} }}$ is the indicator function and equals 1 (resp., 0) if the condition is true (resp., false). Constraints (8a)-(8e) pertain to the task partitioning method, while constraints (8f)-(8i) relate to the parallel scheduling policy. Specifically, (8a) and (8b) describe the task partitioning method. (8c) ensures that each result can be output by one of the task partitions. (8d) and (8e) impose constraints on the number of task partitions, guaranteeing that at least one result is output for each partition. On the other hand, (8f) and (8g) outline the parallel scheduling policy. (8h) specifies the range of values for each task partition scheduling strategy. (8i) ensures that each AD can serve at most one TD for computation. 

Clearly, problem $\mathcal{P}_1$ can be classified as a non-linear integer programming problem. The complexity of the given problem is known to be NP-hard, indicating that it is not feasible to find a solution within polynomial time \cite{43}.

\begin{theorem}
The min-max normalized delay (MMND) problem $\mathcal{P}_1$ is NP-hard.
\end{theorem}

\newenvironment{proof}{{\noindent\it\quad Proof:}}{\hfill $\square$\par}
\begin{proof}
To prove the NP-hardness of problem $\mathcal{P}_1$, we first consider a specific instance of the problem in which the parallel scheduling policy is known, and there is only one TD. Additionally, the number of task partitions is predetermined, and the constraint (8d) is relaxed to $ {\sum\limits_{i = 1}^{{N_{n,res}}} {{{\bf{1}}_{u_{n,i}^t = j}} \ge 0} }$. Consequently, $\mathcal{P}_1$ is transformed into minimizing the normalized delay of the TD through results allocation. This specific instance can be readily reduced to a well-known NP-hard problem, namely the multiple knapsack problem \cite{43}.
\end{proof}

Multiple knapsack problem \cite{44}: The multiple knapsack problem involves a set of items and knapsacks, where each item has a specific profit and volume, and each knapsack has a defined capacity. The objective of the multiple knapsack problem is to select and assign disjoint subsets of items to different knapsacks in order to maximize the total profit. Furthermore, the capacity of each knapsack must be sufficient to accommodate the total volume of the assigned items. In our case, we can consider the $N_{n,res}$ results and $N_{n,part}^t$ task partitions as the items and backpacks, respectively. Thus, filling the items into knapsacks is equivalent to assigning the results to task partitions, aiming to minimize the delay. Given that problem $\mathcal{P}_1$ can be reduced to a multiple knapsack problem, it follows that $\mathcal{P}_1$ is also NP-hard.

\begin{figure}[t]
\centerline{\includegraphics[width=0.5\textwidth]{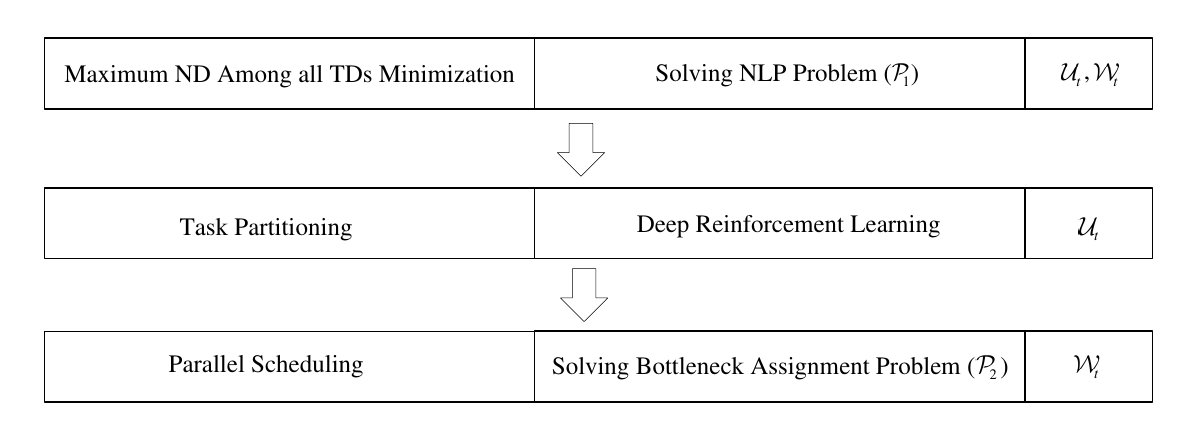}}
\caption{The two-level optimization structure of sloving $\mathcal{P}_1$.}
\label{solution}
\end{figure}

$\mathcal{P}_1$ is a non-linear programming (NLP) problem. However, when $\boldsymbol{\mathcal{U}}_t$ is given, $\mathcal{P}_1$ can be simplified to a bottleneck assignment problem \cite{45}, as described below:
\begin{equation}
\begin{aligned}\label{P}
\mathcal{P}_2:\mathop {\emph{min} }\limits_{{\boldsymbol{{\cal W}}_t}}& \;\mathop {\emph{max} }\limits_{n \in {\cal N}} \;\eta _n^t({\boldsymbol{{\cal H}}^t},\boldsymbol{{\cal F}}_0^t,{\boldsymbol{{\cal F}}_{{{\cal M}}}^t},{\boldsymbol{{\cal U}}_t},{\boldsymbol{{\cal W}}_t}) \\
\mbox{s.t.}&\quad(8f)-(8i).
\end{aligned}
\end{equation}
Accordingly, problem $\mathcal{P}_1$ can be decomposed into two sub-problems: the task partitioning problem and the parallel scheduling problem $\mathcal{P}_2$, as illustrated in Fig. \ref{solution}.

\begin{itemize}
\item{\emph{Task Partitioning (TP)}: The task partitioning problem involves determining the number of partitions for each IoT task and assigning each result to a specific partition. For a given IoT task $\mathcal{T}_n$, there are $N_{n, res}$ possible scenarios for the number of task partitions and ${(N_{n, part}^t)^{{N_{n, res}}}}$ possible scenarios for the assignment of task results when the number of task partitions is $N_{n, part}^t$. Consequently, task $\mathcal{T}_n$ has $\sum\limits_{i = 1}^{{N_{n,res}}} {{{(i)}^{{N_{n,res}}}}} $ distinct partitioning strategies. Therefore, the problem requires exploring among $\prod\limits_{n \in {\cal N}} {\sum\limits_{i = 1}^{{N_{n, res}}} {{{(i)}^{{N_{n, res}}}}}}$ potential partitioning strategies to find an optimal solution. Meta-heuristic search algorithms can utilize heuristic information inherent in the problem to guide the search process, reducing the search space and overall complexity \cite{46}. However, due to the exponentially large search space, these algorithms often require a significant amount of time to converge.}
\item{\emph{Parallel Scheduling (PS)}: Assuming a task partitioning policy is given, the objective of the parallel scheduling problem is to allocate task partitions to different processing locations for parallel execution, aiming to minimize the maximum normalized delay among all TDs. This problem can be regarded as a bottleneck allocation problem \cite{45}. To address this problem, we devise a low-complexity task allocation algorithm that aims to reduce the average delay while ensuring fairness among TDs.}
\end{itemize}

The primary difficulty in solving the problem lies in the TP problem. Traditional optimization algorithms are inherently unsuitable for real-time system optimization in the presence of fast-fading channels and time-varying resource provisioning. To tackle this issue, we propose a deep reinforcement learning-based online partitioning algorithm which can solve the TP problem in millisecond-level average running time.

\subsection{Preliminaries for the Proposed Framework}\label{II-G}

Prior to introducing the solution, this subsection provides the key descriptions and definitions utilized in the proposed framework.

\textbf{Hungarian Algorithm.} The Hungarian algorithm is a combinatorial optimization technique. The core of this algorithm is to find the maximum matching of a bipartite graph by finding augmentation paths [50]. Our study considers task partitions and computing nodes as distinct sets of vertices within a bipartite graph. Consequently, the Hungarian algorithm can be employed to attain the matching between task partitions and computing nodes. Nonetheless, this algorithm cannot ensure fairness among users. Therefore, enhancements are imperative to fulfill the objectives of this investigation.

\textbf{Experience Replay.} Within the realm of DRL, due to the temporal correlation among the samples obtained by the intelligent agent during environmental exploration, these samples do not exhibit independent and identical distribution characteristics. To mitigate the temporal correlation inherent in the samples, a viable approach involves the utilization of an experience replay buffer. This buffer amalgamates historical samples with contemporaneous ones, thereby diminishing data interdependence. Furthermore, experience replay confers the added benefit of rendering the samples reusable, consequently enhancing the efficiency of the learning process.

\textbf{Actor-Critic Algorithm.} The actor-critic algorithm combines the methods of policy gradient (actor) and function approximation (critic). The actor makes action selections grounded in the environmental state and current policy, while the critic module evaluates the action's score based on the actor's chosen action and the prevailing environmental state. Subsequently, the actor module refines its strategy using the score derived from the critic module, and the critic network updates its parameters based on the rewards garnered from the environment using the temporal difference (TD) algorithm. However, the efficacy of this framework hinges on the precision of the critic module's scoring. Furthermore, achieving convergence in the critic network proves challenging, and this challenge compounds with the incorporation of actor updates. To address the convergence dilemma, this study enhances the critic module by leveraging model information to derive precise action assessments. This intervention fosters a more robust and faster convergence of the DRL training process.

\section{Joint TP and PS}
In this section, we propose a low-complexity algorithm to tackle the PS problem. Subsequently, leveraging the solution obtained from the PS problem, we derive an approximately optimal solution for the TP problem.

\subsection{Parallel Scheduling}\label{III-A}
Given a task partitioning strategy $\boldsymbol{\mathcal{U}}_t^0$, the PS problem $\mathcal{P}_2$ can be formulated as a bottleneck allocation problem (BAP) based on Lemma 1.

\begin{lemma}
Under an optimal task partitioning and parallel scheduling policy, the task partitions of each TD are processed in parallel at different locations.
\end{lemma}

\begin{proof}
Assume that $\boldsymbol{\mathcal{U}}_t^*$ and $\boldsymbol{\mathcal{W}}_t^*$ are the optimal task partitioning action and parallel scheduling strategy, where two task partitions of TD $n$ are processed at the same location, i.e., $w_{n, i}^t = w_{n,j}^t, i \ne j$. We can then construct new task partitioning action $\boldsymbol{\mathcal{U}}_t'$ and parallel scheduling strategy $\boldsymbol{\mathcal{W}}_t'$, where $\boldsymbol{\mathcal{U}}_t'$ combines the $i$th and $j$th task partitions of TD $n$ from $\boldsymbol{\mathcal{U}}_t^*$ into a single task partition. The only difference between $\boldsymbol{\mathcal{W}}_t'$ and $\boldsymbol{\mathcal{W}}_t^*$ is the number of task partitions, while the scheduling decision for each task partition remains the same. According to Section \ref{II-B}, the workload and data transfer requirements after merging the two task partitions do not exceed  their cumulative sum. As a result, $ \mathop {\max }\limits_{n \in {\cal N}} \;\eta _n^t({\boldsymbol{{\cal H}}^t},\boldsymbol{{\cal F}}_0^t,{\boldsymbol{{\cal F}}_{{{\cal M}}}^t},{\boldsymbol{{\cal U}}_t}',{\boldsymbol{{\cal W}}_t}') \le \mathop {\max }\limits_{n \in {\cal N}} \;\eta _n^t({\boldsymbol{{\cal H}}^t},\boldsymbol{{\cal F}}_0^t,{\boldsymbol{{\cal F}}_{{{\cal M}}}^t},\boldsymbol{{\cal U}}_t^*,\boldsymbol{{\cal W}}_t^*)$, which contradicts the initial assumption. Therefore, under the optimal task partitioning action and parallel scheduling strategy, the task partitions of each TD are processed in parallel at different locations.
\end{proof}

\begin{theorem}
The PS problem can be classified as a bottleneck allocation problem.
\end{theorem}

\begin{proof}
Considering that $\eta _n^t = (\frac{{L_n^t}}{{T_n^{max}}}>1)?\infty:\frac{{L_n^t}}{{T_n^{max}}}$, as stated in (\ref{eq7}), $\mathcal{P}_2$ can be modified as follows:
\begin{equation}
\begin{aligned}\label{P}
\mathcal{P}_3:\mathop {\emph{min} }\limits_{{\boldsymbol{{\cal W}}_t}} \;&\mathop {\emph{max} }\limits_{n \in {\cal N},i \in \{ 1,\cdots,N_{n,part}^t\} } \eta_{n,i}^t \\
\mbox{s.t.}&\quad(8f)-(8i),
\end{aligned}
\end{equation}
where $\eta_{n,i}^t=(\frac{{L_{n,i}^t}}{{T_n^{max}}}>1)?\infty:\frac{{L_{n,i}^t}}{{T_n^{max}}}$.

Bottleneck allocation problem (BAP) \cite{45}: The BAP involves multiple agents and tasks. Each agent can be assigned to a task, incurring varying costs depending on the agent-task assignment. The objective is to assign exactly one agent to each task in order to minimize the maximum cost among all assignments. By applying Lemma 1 and considering the aforementioned description, it can be readily demonstrated that $\mathcal{P}_3$ qualifies as a BAP.
\end{proof}

To solve $\mathcal{P}_3$, we compute the delay for each task partition when processed at different locations using (1)-(6). Subsequently, we construct a normalized delay matrix $\mathbf{D}^t$ as depicted in Fig. \ref{matrix}. The matrix contains the normalized delay values for task partitions of all TDs when executed at various locations. Specifically, $\eta _{n,i}^{t,l} = (\frac{{L_{n,i}^{t,l}}}{{T_n^{max}}}>1)?\infty:\frac{{L_{n,i}^{t,l}}}{{T_n^{max}}}$ represents the normalized delay when the $i$th task partition of TD $n$ is executed at position $l$ during the $t$th time frame. Here, $l=M+1$ indicates that $\mathcal{T}_{n,i}^t$ is processed on TD $n$, $l=m$ (where $m\in\{1,2,\cdots, M\}$) signifies processing on AD $m$ ($m\in\mathcal{M}$), and $l=0$ indicates processing on ES $0$.

\begin{figure*}[t]
\centerline{\includegraphics[width=0.9\textwidth]{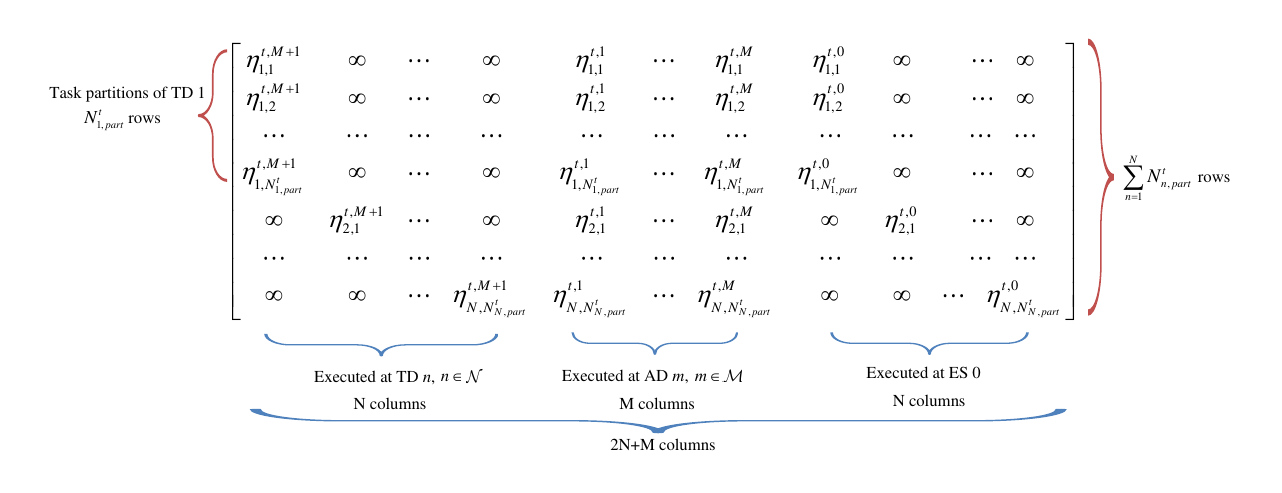}}
\caption{Normalized delay matrix $\mathbf{D}^t$.}
\label{matrix}
\end{figure*}

Fig. \ref{matrix} provides an illustration of matrix $\mathbf{D}^t$, which comprises $\sum\limits_{n = 1}^N {N_{n, part}^t} $ rows and $2N+M$ columns.  Each row corresponds to a task partition, and each column represents a processing location. Rows $\sum\limits_{n = 1}^{i - 1} {N_{n, part}^t + 1} $ to $\sum\limits_{n = 1}^i {N_{n, part}^t}$, where $i\in\{2, 3, \cdots, N\}$, depict the ND of TD $i$'s task partitions when processed at various locations. Columns $1$ to $N$ correspond to TDs, with column $i$ representing TD $i$. As assumed in Section \ref{II-C}, TD $n$'s task partitions cannot be executed on TD $n'$, where $n'\in\mathcal{N}\backslash\{n\}$. Therefore, for TD $n$, all entries in columns $1$ to $N$ (excluding column $n$) of row $n$ are set to $\infty$. Columns $N+1$ to $N+M$ correspond to ADs, with column $N+j$ representing AD $j$. Columns $N+M+1$ to $2N+M$ correspond to the ES $0$, and column $N+M+k$ corresponds to the resources allocated by the ES $0$ to TD $k$. As TD $k$ cannot utilize the resources allocated to TD $k'$ by the ES $0$, where $k'\in\mathcal{N}\backslash\{k\}$, all entries in columns $N+M+1$ to $2N+M$ (excluding column $N+M+k$) of row $k$ are set to  $\infty$.

Based on Lemma 1 and constraint (8i), $\mathcal{P}_3$ can be transformed as follows: minimize the maximum value among the selected positions by choosing one unique column index from each row of matrix  $\mathcal{D}^t$. To solve this problem, we design a fair and delay-minimized task scheduling (FDMTS) algorithm, as presented in Algorithm 1.

\IncMargin{1em}
\begin{algorithm} \SetKwData{Left}{left}\SetKwData{This}{this}\SetKwData{Up}{up} \SetKwFunction{Union}{Union}\SetKwFunction{FindCompress}{FindCompress} \SetKwInOut{Input}{input}\SetKwInOut{Output}{output}
	
	\Input{${\boldsymbol{{\cal H}}^t},\boldsymbol{{\cal F}}_0^t,\boldsymbol{{\cal F}}_{\cal M}^t,{\boldsymbol{{\cal U}}_t}$} 
	\Output{${\boldsymbol{{\cal W}}_t},\emph{max}\; \eta _{n, i}^{t,l}$}
	\BlankLine 

    Initialize $\{ \boldsymbol{{\cal T}}_n^t|n \in {\cal N}\} $ and  $\{ N_{n,part}^t|n \in {\cal N}\} $ based on $\boldsymbol{\mathcal{U}}_t$, construct ${{\bf{D}}^t}$ based on (1)-(6), ${\bf{o}} = []$, and $c = 0$\;
	\For{$i = 1, 2, \cdots, \sum\limits_{n = 1}^N {N_{n,part}^t}$}{ 
        Find the smallest value in row $i$ of ${{\bf{D}}^t}$ and add it to ${\bf{o}}$\;        
    } 
    Select the maximum value from ${\bf{o}}$ and assign it to $c$\;
    Copy ${{\bf{D}}^t}$ to get a copy ${\widehat {\bf{D}}^t}$\;
    Replace all elements of ${\widehat {\bf{D}}^t}$ that are larger than $c$ with $\infty $ to obtain the new matrix ${\overline {\bf{D}} ^t}$\;
    Apply the Hungarian algorithm to ${\overline {\bf{D}} ^t}$ for an efficient solution ${\cal Y} = \{ {m_i}|i \in \{ 1,2,...,\sum\limits_{n\in\mathcal{N}} {N_{n,part}^t} \} \} $, where $m_i$ represents the column index chosen in the $i$th row of ${\overline {\bf{D}} ^t}$\;
    \If{the solution exists}{
        The task scheduling policy $\mathcal{W}_t$ is obtained from $\mathcal{Y}$, and $c$ is the optimal value of $\emph{max}\; \eta _{n, i}^{t,l}$\;
    }
    \Else{ 
            Find the minimum value among the elements in ${{\bf{D}}^t}$ that are larger than $c$, and then update $c$ with that minimum value\;
            \textbf{goto} line 5
        }
    \caption{Task Scheduling Algorithm With Min-Max Fairness Guarantee for PS}\label{algo1} 
    \end{algorithm}
\DecMargin{1em} 

First, we construct a normalized delay matrix based on the current system state and task partitioning action (step 1). Then, we find the minimum values in all rows of this matrix and assign their maximum value to $c$ (steps 2-4). We set all matrix elements greater than $c$ to $\infty$ and apply the Hungarian algorithm to the updated matrix. If a feasible solution exists, the value of $c$ is $\emph{min} \;\;\emph{max}\; \eta _{n, i}^{t,l}$ (steps 5-9). Otherwise, we gradually relax the constraint on the maximum value until a feasible solution is obtained (steps 10-12). The algorithm aims to find the optimal task assignment policy subject to the constraint that the maximum value of all matches is restricted. If a feasible solution cannot be found under the current constraint, we continue to relax the constraint until a feasible solution is found.

\subsection{OTPPS Framework Overview}\label{III-B}

According to Section \ref{III-A}, we can obtain the corresponding parallel scheduling policy for a given task partitioning action. Therefore, it is crucial to determine a near-optimal task partitioning policy in real time. Intuitively, as mentioned in Section \ref{II-F}, we could potentially enumerate all $\prod\limits_{n \in {\cal N}} {\sum\limits_{i = 1}^{{N_{n, res}}} {{{(i)}^{{N_{n, res}}}}} }$ feasible ${\boldsymbol{{\cal U}}_t}$ and select the one that minimizes the objective function $\mathcal{P}_3$. However, such an exhaustive search approach is computationally infeasible, particularly when the problem needs to be frequently resolved with time-varying system states. Additionally, other search-based methods, such as branch-and-bound and Gibbs sampling algorithms \cite{47}, become time-consuming when the values of $N$ and $N_{n, res}$ are large.

Inspired by the work in \cite{48}, we propose a DRL-based framework called OTPPS to address the joint optimization problem considering time-varying channel gains and resource provisioning. We aim to derive a task partitioning policy $\pi$ that can quickly predict the optimal task partitioning action $\boldsymbol{\mathcal{U}}_t^*$ for $\mathcal{P}_1$ once $\boldsymbol{\mathcal{H}}^t$, $\boldsymbol{\mathcal{F}}_0^t$, and $\boldsymbol{\mathcal{F}}_{\mathcal{M}}^t $ are revealed at the beginning of each time frame. Subsequently, we solve the corresponding parallel scheduling problem using Algorithm 1. The task partitioning policy is denoted as
\begin{equation}
\label{eq11}
\pi :\{ {\boldsymbol{{\cal H}}^t},\boldsymbol{{\cal F}}_0^t,\boldsymbol{{\cal F}}_{\cal M}^t\}  \mapsto \boldsymbol{{\cal U}}_t^*.
\end{equation}

The algorithm structure is illustrated in Fig. \ref{framework}. The OTPPS algorithm consists of two alternating stages: task partitioning action and parallel scheduling policy generation, and partitioning policy update. These stages are described as follows.

\begin{figure*}[t]
\centerline{\includegraphics[width=0.8\textwidth]{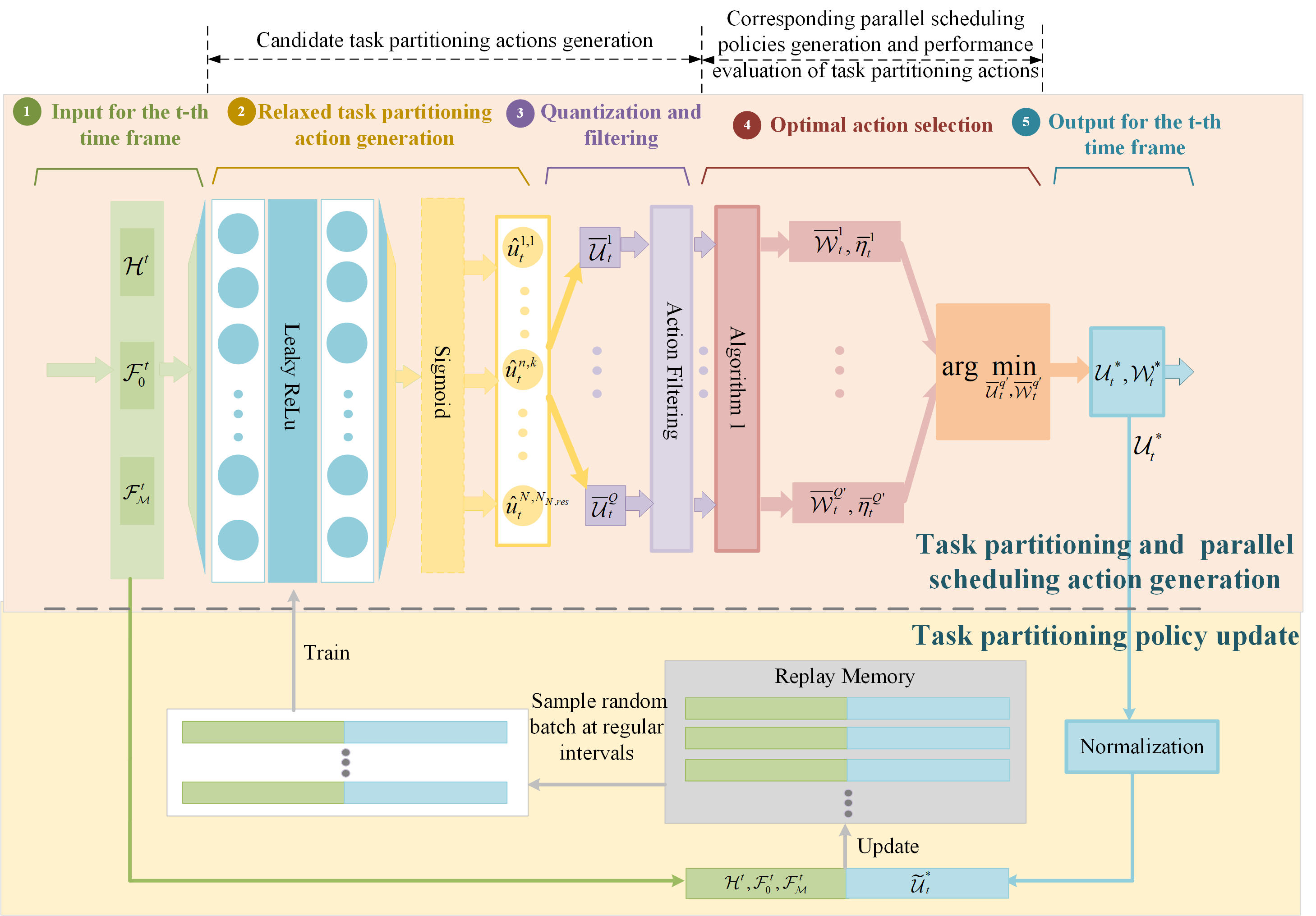}}
\caption{The schematics of the OTPPS framework.}
\label{framework}
\end{figure*}

\begin{enumerate}
\item{\emph{Task Partitioning Action and Parallel Scheduling Policy Generation}: As shown in Fig. \ref{framework}, the action generation phase consists of five steps. The task partitioning action is generated using a DNN, characterized by its embedded parameters $\theta_t$, e.g., the weights that connect the hidden neurons. In the $t$th time frame, the DNN takes $\boldsymbol{\mathcal{H}}^t$, $\boldsymbol{\mathcal{F}}_0^t$, and $\boldsymbol{\mathcal{F}}_{\mathcal{M}}^t$ as inputs and produces a relaxed partitioning action $\hat{\boldsymbol{\mathcal{U}}}_t$, where each entry is relaxed to a continuous value between 0 and 1, based on the current embedded parameters $\theta_t$. The relaxed action is then quantized into $Q$ integer partitioning actions, from which $Q'$ feasible candidate actions are screened out. Among these $Q'$ candidate actions, the objective value of $\mathcal{P}_3$ achievable under each task partitioning action is determined using Algorithm 1. The task partitioning and parallel scheduling actions corresponding to the smallest objective value are selected as the optimal actions $\boldsymbol{\mathcal{U}}_t^*$ and $\boldsymbol{\mathcal{W}}_t^*$, respectively. The intelligent agent then takes $\boldsymbol{\mathcal{U}}_t^*$ and $\boldsymbol{\mathcal{W}}_t^*$ and adds the newly obtained state-action pair $(\{{\boldsymbol{{\cal H}}^t},\boldsymbol{{\cal F}}_0^t,\boldsymbol{{\cal F}}_{\cal M}^t\}, \widetilde {\boldsymbol{\cal U}}_t^*)$ to the experience replay buffer, where $\widetilde {\boldsymbol{\cal U}}_t^*$ represents the normalized value of $\boldsymbol{{\cal U}}_t^*$.}
\item{\emph{Task Partitioning Policy Update}: The task partitioning actions obtained in the action generation stage are used to update the parameters of the DNN through the partitioning policy update stage. As illustrated in Fig. \ref{framework}, we employ an experience replay buffer with limited capacity to store past state-action pairs. In each time frame, the tuple $(\{{\boldsymbol{{\cal H}}^t},\boldsymbol{{\cal F}}_0^t,\boldsymbol{{\cal F}}_{\cal M}^t\},\widetilde {\boldsymbol{\cal U}}_t^*)$ obtained from the action generation stage is added to the buffer as a new training data sample. When the buffer is full, the newly generated data sample replaces the oldest one. During the policy update stage, a batch of training samples is randomly selected from the buffer to train the DNN, resulting in the update of its parameters from $\theta_t$ to $\theta_{t+1}$. The updated DNN is then utilized in the subsequent time frame to generate the task partitioning action $\boldsymbol{\mathcal{U}}_{t+1}^*$ and the corresponding parallel scheduling policy $\boldsymbol{\mathcal{W}}_{t+1}^*$ based on the new system status ${\rm{\{ }}{\boldsymbol{{\cal H}}^{t+1}},\boldsymbol{{\cal F}}_0^{t+1},\boldsymbol{{\cal F}}_{\cal M}^{t+1}{\rm{\} }}$.}

\end{enumerate}

Such iterations are repeated as new time frames arrive, gradually improving the policy $\pi_{\theta_t}$ of the DNN. The detailed descriptions of the two stages are provided in the following section. It is important to note that these two stages operate on different time scales. The action generation stage is executed at the beginning of each time frame, while the policy update stage is performed every $\Omega$ time frames. Additionally, we initiate the training process when the number of samples in the buffer exceeds half of its capacity, ensuring a sufficient number of new data samples for effective training.

\section{The OTPPS Framework}

\subsection{Action Generation}\label{IV-A}
Let us consider the system status realization $\{{\boldsymbol{{\cal H}}^t},\boldsymbol{{\cal F}}_0^t,\boldsymbol{{\cal F}}_{\cal M}^t\}$ observed in the $t$th time frame, where $t=1, 2, \cdots$. The parameters of the DNN $\theta_t$ are randomly initialized following a zero-mean normal distribution when $t=1$. In this study, we employ the Rectified Linear Unit (ReLU) as the activation function for the hidden layers and the Sigmoid function for the output layer. Consequently, the DNN generates a relaxed partitioning action denoted as $\hat{\boldsymbol{\mathcal{U}}}_t$, represented by a parameterized function $\hat{\boldsymbol{\mathcal{U}}}_t=f_{\theta_t}(\boldsymbol{\mathcal{H}}^t,\boldsymbol{\mathcal{F}}_0^t, \boldsymbol{\mathcal{F}}_{\mathcal{M}}^t)$, where
\begin{equation}
\label{eq12}
{\widehat {\boldsymbol{\cal U}}_t} = \{ \hat u_t^{n,k}|\hat u_t^{n,k} \in [0,1],n \in {\cal N},k \in \{ 1,\cdots,{N_{n,res}}\} \}
\end{equation}
and $\hat u_t^{n,k}$ corresponds to the $k$th result of TD $n$.

To generate an integer partitioning action that satisfies condition (8c), we quantize $\hat{\boldsymbol{\mathcal{U}}}_t$ into $Q$ candidate partitioning actions. Among these actions, we screen out the ones with the number of task partitions not exceeding $2N+M$ based on Theorem 3, resulting in $Q'$ remaining candidate actions. The performance of each candidate action is evaluated using Algorithm 1. The candidate action with the lowest maximum normalized delay (MND) is selected as the final task partitioning action, along with its corresponding parallel scheduling policy. It is worth noting that an effective quantization method allows for the generation of only a few candidate actions, thereby reducing the computational complexity. Additionally, the quantized actions derived from the relaxed action should exhibit sufficient diversity to achieve a lower MND.

\begin{theorem}
The optimal task partitioning action ensures that the number of task partitions does not exceed $2N+M$.
\end{theorem}

\begin{proof}
Assuming that the number of task partitions generated by the optimal task partitioning action exceeds $2N+M$. Referring to the situation depicted in Fig. \ref{matrix}, the number of rows in the normalized delay matrix would be greater than the number of columns. Consequently, during the parallel scheduling process, there would be at least two task partitions executed in the same location, which contradicts Lemma 1. Hence, the assumption is proven false, and Theorem 3 is established.
\end{proof}

\IncMargin{1em}
\begin{algorithm} \SetKwData{Left}{left}\SetKwData{This}{this}\SetKwData{Up}{up} \SetKwFunction{Union}{Union}\SetKwFunction{FindCompress}{FindCompress} \SetKwInOut{Input}{input}\SetKwInOut{Output}{output}
	
	\Input{$\hat{\boldsymbol{\mathcal{U}}}_t, Q$} 
	\Output{The set of Q quantized task partitioning actions $\{ \bar {\boldsymbol{\cal U}}_t^1,\bar {\boldsymbol{\cal U}}_t^2, \cdots,\bar {\boldsymbol{\cal U}}_t^Q\}$}
	\BlankLine 

    Initialize $\{ \bar {\boldsymbol{\cal U}}_t^1,\bar {\boldsymbol{\cal U}}_t^2, \cdots ,\bar {\boldsymbol{\cal U}}_t^Q\}$ with 0, where $\bar u_t^{q,n,k} = 0$ represents that the $k$th result of task $\mathcal{T}_n$ in the $q$th quantized action will be computed by the virtual partition (partition 0)\; 
	\For{$q=1, 2, \cdots, Q$}{
        \For{$n= 1, 2, \cdots, N$}{
            Create $N_{n,res}$ empty vectors, denoted as ${\bf{r}}_1$, ${\bf{r}}_2$, $\cdots$, ${\bf{r}}_{N_{n,res}}$. Also create a set ${\cal I}_n = \{1, 2, \cdots, N_{n,res}\}$ to represent the set of indices for results of ${\cal T}_n$\;
            \For{$j = 1, 2, \cdots, {N_{n,res}} - 1$}{
                \For{$k = 1, 2, \cdots, {N_{n,res}}$}{
                    \If{$\frac{{j - 1}}{{{N_{n,res}}}} + \frac{{q - 1}}{{Q{N_{n,res}}}} \le \hat u_t^{n,k} < \frac{j}{{{N_{n,res}}}} + \frac{{q - 1}}{{Q{N_{n,res}}}}$}{
                        Add $k$ to $\textbf{r}_j$\;
                        Move $k$ out of $\mathcal{I}_n$\;  
                    }
                }
            }
            Add the remaining elements in $\mathcal{I}_n$ to $\textbf{r}_{N_{n,res}}$\;
            \For{$j = 1, 2, \cdots, {N_{n,res}}$}{
                \If{$\textbf{r}_j$ is not empty}{
                    \For{$k \in {{\bf{r}}_j}$}{
                        $\bar u_t^{q,n,k} \leftarrow \emph{min}(\textbf{r}_j)$
                    }
                }
            }
            Sort all results of $\mathcal{T}_n$ in ascending order by their corresponding values $\bar u_t^{q,n,i}$, and then renumber them according to the ranking\;
        }
                
    } 
    \caption{Sliding Threshold Quantization (STQ) Algorithm}\label{algo2} 
    \end{algorithm}  
\DecMargin{1em} 

The order-preserving quantization method was originally introduced in \cite{48} to explore the output of the DNN. Its core principle is to preserve the ordering of vector entries before and after quantization. In \cite{49}, the GNOP method was proposed, which combines Gaussian noise-added approach and the order-preserving quantization method. This technique introduces noise to better explore the action space and increase the likelihood of finding a local optimum within a large action space. However, both methods are specifically designed for binary offloading actions and are not suitable for addressing the problem under investigation. Here, we devise the sliding threshold quantization (STQ) algorithm, as depicted in Algorithm 2, to achieve efficient and low-complexity action quantization. It ensures highly diverse quantized actions.

The algorithm consists of two main parts: result classification and task partition numbering. Initially, each component of the $Q$ quantized partitioning actions is assigned a value of 0 (step1). Subsequently, we iterate $Q$ times to reassign values to each component of the $Q$ quantized actions. During each iteration, we determine the quantization threshold for each task, enabling the classification of all the results associated with each task (Steps 2-10). Finally, we assign numerical labels to the results of each task based on their respective categories (Steps 11-15). 

Fig. \ref{STQ} provides an example to illustrate the sliding threshold quantization method. In this particular scenario, there are two tasks, denoted as $\mathcal{T}_1$ and $\mathcal{T}_2$. $\mathcal{T}_1$ requires to output four results, and $\mathcal{T}_2$ requires to output three. Since the number of task partitions cannot exceed the number of output results, we initially set the thresholds for these tasks as $\{\frac{1}{4}, \frac{2}{4}, \frac{3}{4}, 1\}$ and $\{\frac{1}{3}, \frac{2}{3}, 1\}$, respectively. Assuming that $Q$ is set to 4, we can calculate the sliding threshold steps for both tasks, which are $\frac{1}{4\times 4}$ and $\frac{1}{4\times 3}$. When the DNN outputs ${\widehat {\boldsymbol{\cal U}}_t} = \{ 0.2,0.4,0.7,0.9,0.3,0.7,0.9\} $, applying the STQ algorithm generates four quantized actions: $\bar {\boldsymbol{\cal U}}_t^1 = \{ 1,2,3,4,1,2,2\}$, $\bar {\boldsymbol{\cal U}}_t^2 = \{ 1,2,3,4,1,2,3\}$, $\bar {\boldsymbol{\cal U}}_t^3 = \{ 1,2,3,4,1,2,3\}$, and $\bar {\boldsymbol{\cal U}}_t^4 = \{ 1,1,2,2,1,2,2\}$, respectively. It can be seen that $\bar {\boldsymbol{\cal U}}_t^2 $ and $\bar {\boldsymbol{\cal U}}_t^3$ are identical. Consequently, the subsequent filtering operation will remove infeasible or duplicate actions, reducing the computational effort for action evaluation.

\begin{figure}[t]
\centerline{\includegraphics[width=0.5\textwidth]{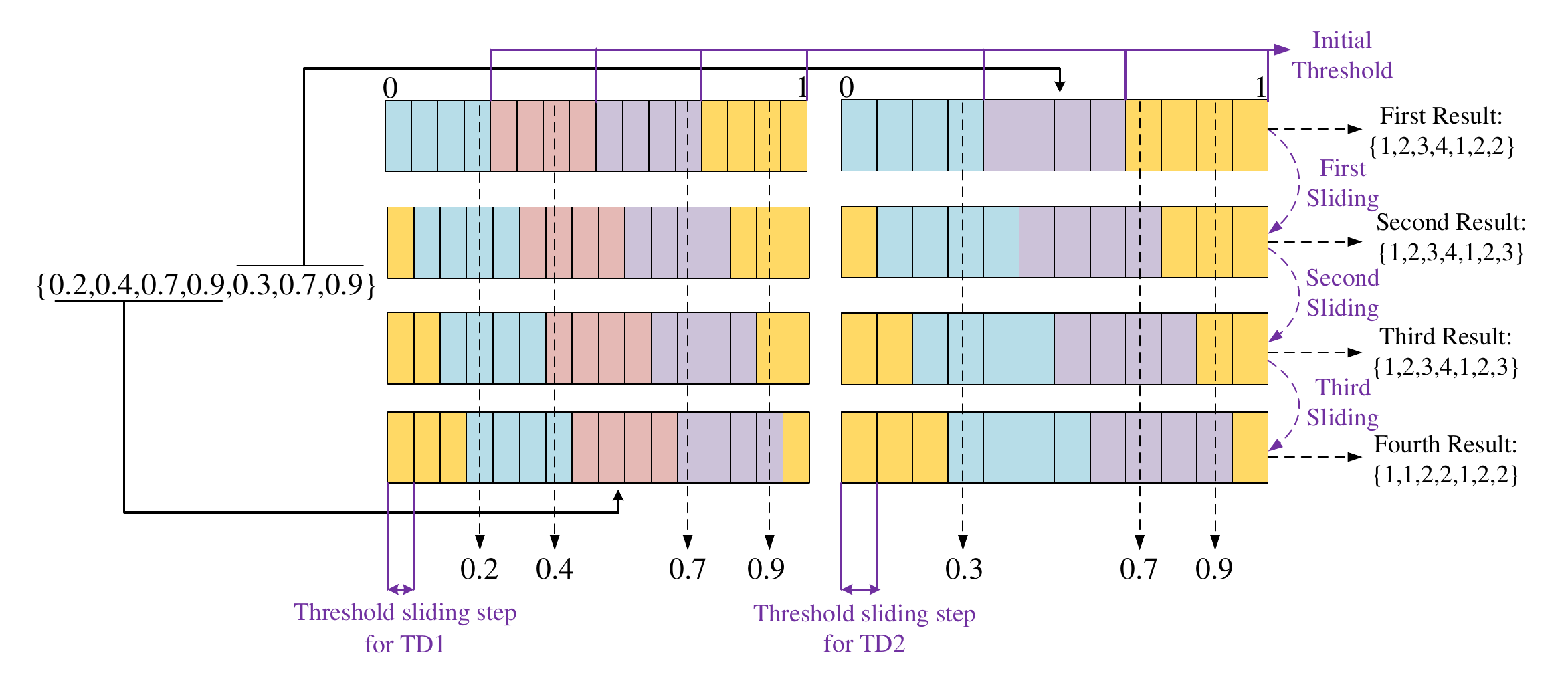}}
\caption{An illustration of the execution process of STQ algorithm.}
\label{STQ}
\end{figure}

K-Means is an alternative quantization method that involves selecting a value, denoted as $K$, to determine the number of task partitions. The results are then partitioned into $K$ categories to complete the quantization process. Unlike STQ, this method generates a set of candidate actions where the only difference lies in the number of task partitions. However, for the same number of task partitions, there are typically various partitioning actions, as depicted in ways $S2$, $S3$, and $S4$ of Fig. \ref{task}(b). Consequently, STQ exhibits a higher diversity in the set of candidate actions, leading to an increased likelihood of finding a local optimum.  In Section \ref{V-B}, it will be demonstrated that the proposed STQ outperforms the K-Means method.

Recall that each feasible candidate action $\overline{ \boldsymbol{\cal U}} _t^{q'}$ can achieve $\overline \eta  _t^{q'}({{\boldsymbol{\cal H}}^t},{\boldsymbol{\cal F}}_0^t,{\boldsymbol{\cal F}}_{\cal M}^t,\overline {\boldsymbol{\cal U}} _t^{q'})$ by solving $\mathcal{P}_3$. Thus, the best task partitioning action at the $t$th time frame is selected as follows:
\begin{equation}
\label{eq13}
{\boldsymbol{\cal U}}_t^* = \arg \;\mathop {\min }\limits_{\overline {\boldsymbol{\cal U}} _t^i \in \{ \overline {\boldsymbol{\cal U}} _t^{q'}\} } \overline \eta  _t^{q'}({\boldsymbol{{\cal H}}^t},\boldsymbol{{\cal F}}_0^t,\boldsymbol{{\cal F}}_{\cal M}^t,\overline {\boldsymbol{\cal U}} _t^{q'}).
\end{equation}
Note that the evaluation of $\overline \eta  _t^{q'}({\boldsymbol{{\cal H}}^t},\boldsymbol{{\cal F}}_0^t,\boldsymbol{{\cal F}}_{\cal M}^t,\overline {\boldsymbol{\cal U}} _t^{q'})$ for $Q'$ times can be processed in parallel to accelerate the computation of Equation (\ref{eq13}). Subsequently, the intelligent agent outputs the task partitioning action ${\boldsymbol{\cal U}}_t^*$ along with its corresponding optimal parallel scheduling policy ${\boldsymbol{\cal W}}_t^*$.

\subsection{Task Partitioning Policy Update}\label{IV-B}
The task partitioning action obtained in Equation (\ref{eq13}) is used to update the task partitioning policy of the DNN. Specifically, we maintain an initially empty experience replay buffer with limited capacity. At the $t$th time frame, the task partitioning action $\boldsymbol{\mathcal{U}}_t^*$ obtained during the action generation phase is normalized to scale the value of each entry between 0 and 1, resulting in $\widetilde {\boldsymbol{\cal U}}_t^*$. Then, a new training data sample $(\{ {{\boldsymbol{\cal H}}^t},{\boldsymbol{\cal F}}_0^t,\boldsymbol{{\cal F}}_{\cal M}^t\} ,\widetilde {\boldsymbol{\cal U}}_t^*)$ is added to the buffer. When the buffer reaches its full capacity, the newly generated data sample replaces the oldest one.

We continue to use the example shown in Fig. \ref{STQ} to illustrate the normalization method. Since $\mathcal{T}_1$ produces four results and $\mathcal{T}_2$ produces three results, the maximum values of each component in the partitioning action for the two tasks are 4 and 3, respectively. Therefore, considering $\bar{\boldsymbol{\mathcal{U}}}_t^1 = \{ 1,2,3,4,1,2,2\} $ as an example, it is normalized to obtain $\widetilde {\boldsymbol{\cal U}}_t^1 = \{ \frac{1}{4},\frac{1}{2},\frac{3}{4},1,\frac{1}{3},\frac{2}{3},\frac{2}{3}\}$.

The data samples stored in the buffer are used to train the DNN. Specifically, in the $t$th time frame, we randomly select a batch of training data samples $\{ (\{{{\boldsymbol{\cal H}}^i},{\boldsymbol{\cal F}}_0^i,{\boldsymbol{\cal F}}_{\cal M}^i\},\widetilde {\boldsymbol{\cal U}}_i^*)| i \in {{\cal I}_t}\}$ from the buffer, where ${{\cal I}_t} $ represents the set of chosen time indices. Then, the parameters $\theta_t$ of the DNN are updated by applying the Adam algorithm to reduce the mean square error loss, as
\begin{equation}
\label{eq14}
Loss({\theta _t}) = \frac{1}{{|{{\cal I}_t}|}}\sum\limits_{i \in {{\cal I}_t}} {||{f_{{\theta _t}}}(} {{\boldsymbol{\cal H}}^i},{\boldsymbol{\cal F}}_0^i,{\boldsymbol{\cal F}}_{\cal M}^i) - \widetilde {\boldsymbol{\cal U}}_i^*)||_2^2,
\end{equation}
where $|\mathcal{I}_t|$ denotes the size of $\mathcal{I}_t $, and $||{\bf{a}}||_2^2$ represents the summation operation over the square of each element in the vector $\textbf{a}$. For brevity, the detail of the Adam algorithm is omitted here. In practice, we initiate the training process when the number of samples exceeds half of the buffer size, and the DNN is trained every $\Omega$ frames to accumulate a sufficient number of new data samples in the buffer.

\IncMargin{1em}
\begin{algorithm}
\SetKwData{Left}{left}\SetKwData{This}{this}\SetKwData{Up}{up} \SetKwFunction{Union}{Union}\SetKwFunction{FindCompress}{FindCompress} \SetKwInOut{Input}{input}\SetKwInOut{Output}{output}
	
	\Input{${{\boldsymbol{\cal H}}^t},{\boldsymbol{\cal F}}_0^t,{\boldsymbol{\cal F}}_{\cal M}^t$} 
	\Output{${\boldsymbol{\cal U}}_t^*,{\boldsymbol{\cal W}}_t^*$}
	\BlankLine 

        Initialize DNN with random parameters $\theta_1$ following a zero-mean normal distribution and empty buffer. Set iteration number $M$, the buffer capacity $L$, and the training interval $\Omega$\;
	\For{$t = 1, 2, \cdots, I_{max}$}{ 
            Input ${{\boldsymbol{\cal H}}^t},{\boldsymbol{\cal F}}_0^t,{\boldsymbol{\cal F}}_{\cal M}^t$ into the DNN to generate a relaxed task partitioning action ${\widehat {\boldsymbol{\cal U}}_t} = {f_{{\theta _t}}}({{\boldsymbol{\cal H}}^t},{\boldsymbol{\cal F}}_0^t,{\boldsymbol{\cal F}}_{\cal M}^t)$\;
            Quantize ${\widehat {\boldsymbol{\cal U}}_t}$ into Q integer task partitioning actions $\{ \bar{\boldsymbol{\mathcal{U}}} _t^q\} $ according to \textbf{Algorithm 2}\;
            Filter out repetitive actions and infeasible actions, leaving $Q'$ feasible candidates $\{ \bar{\boldsymbol{\cal U}}_t^{q'}\} $\;
            The $Q'$ candidate actions are evaluated in parallel using \textbf{Algorithm 1} to obtain the maximum value of the normalized delay achievable for each task partitioning action and the corresponding parallel scheduling policy $\overline {\boldsymbol{\cal W}} _t^{q'}$\;
            Select the optimal task partitioning action according to (\ref{eq13}) and output it and the corresponding parallel scheduling strategy\;
            Normalize each element of ${\boldsymbol{\cal U}}_t^*$ to get $\widetilde {\boldsymbol{\cal U}}_t^*$\;
            Update the replay buffer by adding  $(\{{{\boldsymbol{\cal H}}^t},{\boldsymbol{\cal F}}_0^t,\boldsymbol{{\cal F}}_{\cal M}^t\},\widetilde {\boldsymbol{\cal U}}_t^*\})$ and replacing the oldest data sample if the buffer is full\;
            \If{$t>L/2$ and $t$ mod $\Omega$=0}{
                Randomly sample a batch of data set $\{ (\{{{\boldsymbol{\cal H}}^i},{\boldsymbol{\cal F}}_0^i,{\boldsymbol{\cal F}}_{\cal M}^i\},\widetilde {\boldsymbol{\cal U}}_i^*)|i \in {{\cal I}_t}\}$ from the buffer\;
                Train the DNN with $\{ (\{{{\boldsymbol{\cal H}}^i},\boldsymbol{{\cal F}}_0^i,{\boldsymbol{\cal F}}_{\cal M}^i\},\widetilde {\boldsymbol{\cal U}}_i^*)|i \in {{\cal I}_t}\}$ and update $\theta_t$ using the Adam algorithm\;
            }
    } 
    
    \caption{Online Task Partition Action and Parallel Scheduling Policy Generation (OTPPS) Algorithm}\label{algo3} 
    \end{algorithm}

\DecMargin{1em} 

Overall, the DNN iteratively learns from the state-action pairs $\{ (\{{{\boldsymbol{\cal H}}^i},{\boldsymbol{\cal F}}_0^i,{\boldsymbol{\cal F}}_{\cal M}^i\},\widetilde {\boldsymbol{\cal U}}_i^*)| i \in {{\cal I}_t}\}$ and generates improved task partitioning actions as time progresses. Moreover, due to the finite buffer size constraint, the DNN exclusively learns from the most recent data samples generated by the most recent (and more refined) task partitioning policies. This closed-loop reinforcement learning mechanism continually enhances its task partitioning policy until convergence. The pseudo-code for the OTPPS algorithm is provided in Algorithm 3.

\subsection{Computational Complexity and Convergence Analysis}\label{IV-C}
This subsection analyzes the computational complexity and convergence of the proposed algorithm.

First, based on the above algorithms, Theorem 4 analyzes the time complexity of the OTPPS algorithm as follows.

\begin{theorem}
The time complexity of the OTPPS algorithm is $ \mathcal{O}(ZN^2M^2)$, where $Z$ denotes the total number of IoT devices.
\end{theorem}

\begin{proof}
The environmental state, parameterized in nature, is initially fed into the DNN to generate a relaxed task partitioning action, which requires $\mathcal{O}(1)$. Subsequently, Alg.2 is invoked to produce a set of $Q$ quantized task partitioning actions. Following this, each task partitioning action undergoes an evaluation to ascertain its adherence to Theorem 3, thus sieving out a subset of $Q'$ candidate task splitting actions. This validation step requires $\mathcal{O}(\mathcal{Q})$. These candidate task partitioning actions are concurrently submitted to Alg.3 for assessment, deriving the corresponding parallel scheduling strategy and the minimum attainable maximum normalized delay. Ultimately, the outputted task partitioning alongside its corresponding parallel scheduling policy is determined based on Equation (13), which needs at most $\mathcal{O}(\mathcal{Q})$.

The computational complexity of Alg.2 can be easily analyzed as $ \mathcal{O}(Q\sum_{n=1}^N(2N_{n, res}-1))$.

In Alg.3, the Hungarian algorithm is used in the inner layer. According to \cite{50}, its computational complexity is $\mathcal{O}(VE) $, where $V$ signifies the count of vertices and $E$ symbolizes the count of edges. In this study, referring to Fig. 6, it's evident that $ V=\sum_{n=1}^NN_{n, part}^t+M+2N$ and $ E=2\sum_{n=1}^NN_{n, part}^t+\sum_{n=1}^NN_{n, part}^tM=\sum_{n=1}^NN_{n, part}^t(2+M)$. Thus, the computational complexity for every iteration in Alg. 2 materializes as $ \mathcal{O}((\sum_{n=1}^NN_{n, part}^t+M+2N)\cdot\sum_{n=1}^NN_{n, part}^t(2+M))$. Additionally, the utmost number of iterations shouldn't surpass $E$. Consequently, the time complexity of Alg.3 is $ \mathcal{O}((\sum_{n=1}^NN_{n,part}^t+2N+M)\cdot(\sum_{n=1}^NN_{n,part}^t(2+M))^2)$.

Summing up the above analyses, the computational complexity of the OTPPS algorithm can be expressed as $ \mathcal{O}(1)+\mathcal{O}(Q\sum_{n=1}^{N}(2N_{n, res}-1))+\mathcal{O}(Q)+ \mathcal{O}((\sum_{n=1}^NN_{n, pant}^t+2N+M)\cdotp(\sum_{n=1}^NN_{n, pant}^t(2+M))^2)+\mathcal{O}(Q) $. Upon simplification, the computational complexity of the OTPPS algorithm can be approximated as $ \mathcal{O}_{OTTPS}\approx\mathcal{O}(ZN^2M^2)$.
\end{proof}

The convergence of the OTPPS algorithm is analyzed next.

As shown in Fig.\ref{framework}, the OTPPS framework operates on the foundation of an actor-critic learning structure. Within this framework, the critic module accomplishes a rapid and precise assessment of the task partitioning action's performance by employing Alg.2. Therefore, there is no convergence problem with the critic module of OTPPS. On the other hand, the actor network in this framework remains consistent with its counterpart in traditional actor-critic DRL, directly producing actions as a response to the input environmental state. The convergence of the actor-critic DRL has been demonstrated in \cite{51}, underscoring the convergent nature of the OTPPS algorithm proposed in this study. Additionally, the OTPPS algorithm introduces an additional quantization and filtering procedure to the actor network's output. This augmentation generates a limited set of candidate task partitioning actions, strategically amplifying the likelihood of discovering local optima and thereby expediting the convergence of the actor network.

\section{Performance Evaluation}
In this section, we conduct extensive simulations to evaluate the performance of the OTPPS framework. First, the simulation settings are given. Then the experimental results are presented to demonstrate the effectiveness of the OTPPS framework in minimizing the average delay for all users and ensuring fairness among users in the device-assisted MEN.

\begin{table}[t]
\centering
\caption{Main Simulation Parameters}
\label{param}
\resizebox{0.8\columnwidth}{!}{%
\begin{tabular}{@{}l|l@{}}
\toprule
\textbf{Parameters}\qquad \qquad \qquad \qquad& \textbf{Value}            \\ \midrule
$f_n$               & 0.4-0.6GHz uniformly      \\
$f_m^t$             & 0.2-1.6GHz uniformly      \\
$f_{0,n}^t$         & 1.0-2.0GHz uniformly      \\
$p_n$               & 0.1W                      \\
$p_0$               & 1.5W                      \\
$W$                 & 5MHz                      \\
$N_0$               & $10^{-10}$W                      \\
$Z_n$               & 300-500KB uniformly       \\
$D_n$               & 30-80Megacycles uniformly \\
$I_{max}$           & 10000                     \\
$\Omega$            & 10                        \\
$|\mathcal{I}_t|$   & 128                       \\ \bottomrule
\end{tabular}%
}
\end{table}

\subsection{Simulation Setting}\label{V-A}
Similar to \cite{30}, we consider a square region with sides of 200 meters that comprises multiple task IoT devices, multiple auxiliary IoT devices with idle computation resources, and an SBS equipped with an ES. We assume that the average channel gain $\bar{h}_{i}$ follows a path-loss model $\bar{h}_{i}=A_{d}\left(\frac{3\times10^{8}}{4\pi f_{c}d_{i}}\right)^{d_{e}}$,  where $A_{d}=4.11$ denotes the antenna gain \cite{48}, $ f_{c}=915$ MHz denotes the carrier frequency, $ d_{e}=3$ denotes the path loss exponent \cite{40}, and $d_i$ in meters denotes the distance between the $i$th IoT device and the ES. $h_i$ denotes the channel gain between the ith IoT device and the BS, which follows an i.i.d. Rician distribution with line-of-sight link gain equal to $0.3\bar{h}_i$ \cite{40}. Furthermore, we consider that the uplink and downlink experience the same fading on the assumption of channel reciprocity \cite{48}. We also refer to \cite{19} for some experiment parameters. The key simulation parameters used in this study are listed in Table \ref{param}. All the simulations are performed on a Pytorch 1.10.2 platform with an Intel Core i7-13650HX 4.9GHz CPU and 16 GB of memory.

Here, we set the $T^{max}_n$ for all TDs to be one unit of time, ensuring that the normalized delay of each TD is equal to the task completion time. Therefore, for convenience, we directly use the user's task completion time as our optimization objective in the following.

For task partitioning, we establish the following baseline algorithms to assess the performance of the OTPPS algorithm. Furthermore, given that the proposed framework represents an enhanced DRL algorithm employing an actor-critic structure, we additionally choose DQN \cite{53} and conventional actor-critic \cite{30} as comparative algorithms. It’s worth noting that DQN is a value-based DRL method, while the actor-critic algorithm is a policy-based DRL method.

\begin{enumerate}
\item{\emph{NOSP}: All users' tasks are not split, and each task is scheduled as a whole.}
\item{\emph{MINGRA}: Minimum granularity partitioning, where each task is divided based on the maximum number of task partitions.}
\item{\emph{K-Means} \cite{55}}: The relaxed task partitioning actions output by the DNN are quantized using the K-Means clustering method.
\item{\emph{Exhaustive}: All possible partitioning actions are iterated, and the optimal one is selected.}
\end{enumerate}

It is important to note that the task partitions generated by the aforementioned baseline algorithms are scheduled using the FDMTS algorithm (Algorithm 1) proposed in this paper.

For task scheduling, similar to \cite{30}, we employ four benchmark algorithms:
\begin{enumerate}
\item{\emph{Local}: All tasks of all users are executed on ES $0$.}
\item{\emph{Greedy} \cite{54}: Each task partition is greedily selected to be executed on the ES $0$, TD $n$, or AD $m$ based on its estimated finish time.}
\item{\emph{Random}: The scheduling strategy for each task partition is selected randomly.}
\item{\emph{Kuhn-Munkres} \cite{16}}: Task partitions are matched to processing locations using the Kuhn-Munkres algorithm in graph theory.
\end{enumerate}

It is important to note that the task partitions used for scheduling in the above baseline algorithms are generated using the DRL-based approach proposed in this paper.

Next, we will evaluate the performance of our algorithm in terms of the following metrics:

\emph{1) Convergence Performance:} The convergence speed in DRL algorithms holds paramount significance. Algorithms exhibiting rapid convergence display enhanced robustness, reduced susceptibility to noise and instability, and improved generalization capabilities.

\emph{2) Average Algorithm Execution Time:} The algorithm’s execution time denotes the mean interval between inputting system state parameters and producing decisions. This metric offers an intuitive measure of the algorithm’s computational intricacy. A reduced average execution time implies expedited decision-making and response generation by the controller, thereby enhancing the real-time performance of the system.

\emph{3) Average Task Processing Delay:} Here, the average delay of task processing refers to the average normalized delay of all tasks in the same time frame. This metric corresponds to system-level performance, and lower values correspond to better QoE and system-level performance.

\emph{4) Fairness Index:} The fairness index is an important metric used to measure fairness among users. As described in reference [51], we adopt the Jain fairness index, which is defined as
\begin{equation}
\label{eq15}
F = \frac{{{{(\sum\limits_{n \in {\cal N}} {L_n^t} )}^2}}}{{N\mathop {\sum\limits_{n \in {\cal N}} {{{(L_n^t)}^2}} }}},
\end{equation}
where $\mathcal{N}$ represents the set of TDs. Jain’s fairness index $F$ ranges from $1/N$ to $1$, with the maximum value achieved when all TDs have equal normalized task processing delay. A higher value of $F$ indicates a higher level of fairness in the scheme.

\begin{figure*}[t]
\centering
\subfloat[]{\includegraphics[width=0.33\textwidth]{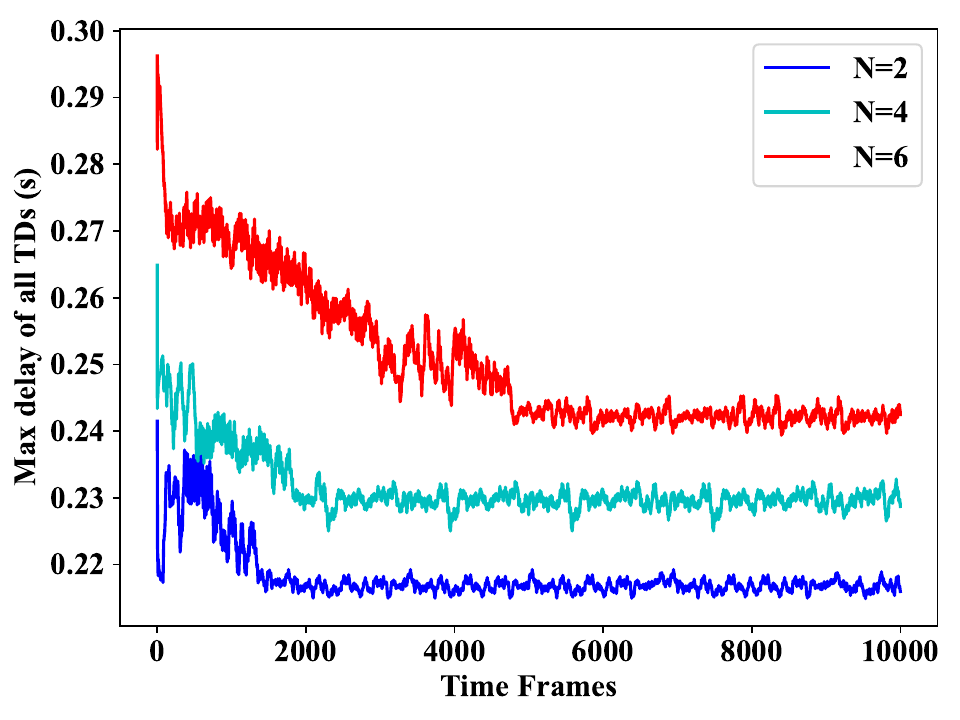}%
\label{fig_first_case}}
\subfloat[]{\includegraphics[width=0.33\textwidth]{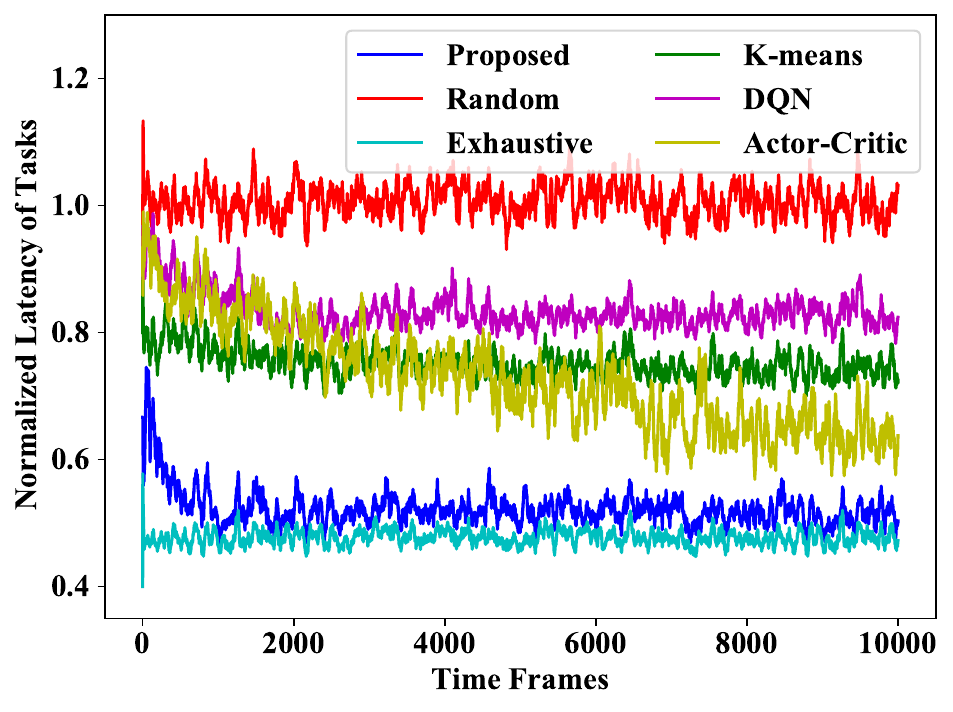}%
\label{fig_second_case}}
\subfloat[]{\includegraphics[width=0.33\textwidth]{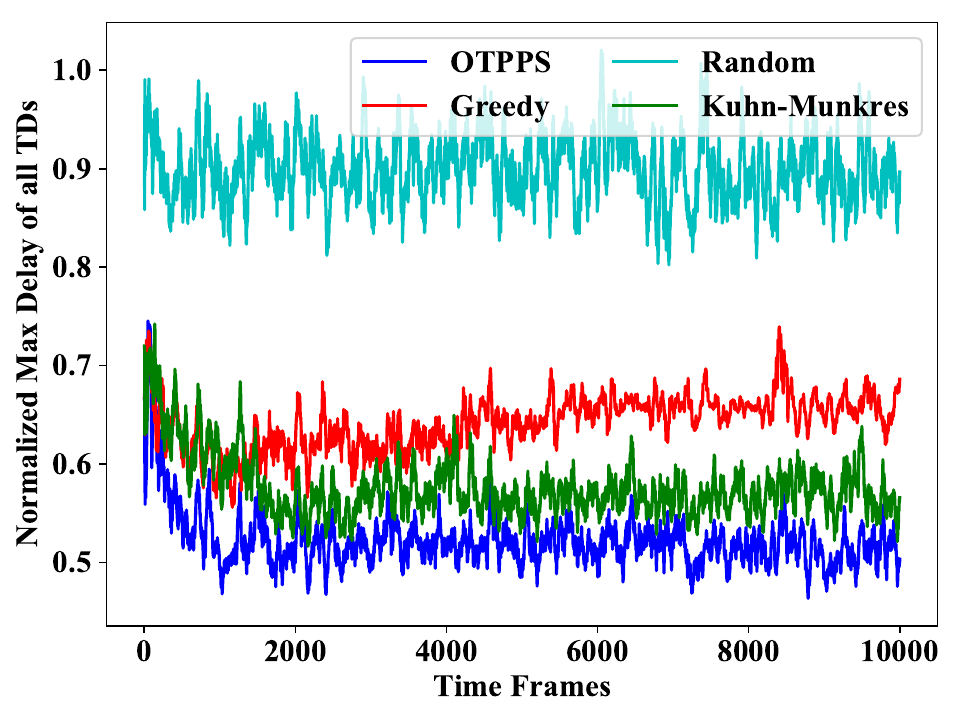}%
\label{fig_second_case}}
\caption{Convergence performance of different algorithms. (a) Comparison of convergence performance of OTPPS algorithm with different numbers of tasks. (b) Comparison of convergence performance of different task partitioning algorithms. (c) Comparison of the convergence performance of different task scheduling algorithms.}
\label{fig10}
\end{figure*}

\subsection{Simulation Results}\label{V-B}

\emph{1) Convergence Performance:} Fig. \ref{fig10} illustrates the convergence performance of the different algorithms. As depicted in Fig. \ref{fig10}(a), the proposed algorithm achieves rapid convergence in terms of the maximum delay of all tasks. The fluctuation of the learning curves over time is attributed to various random factors present in the generated data, including fast-fading channels and time-varying resource provisioning. For $N = 2$, $N = 4$, and $N = 6$, the convergence is achieved in approximately 1600, 2200, and 4600 frames, respectively. Consequently, it can be predicted that the number of iterations required for convergence will gradually increase with the number of TDs. Additionally, it is evident from the figure that the proposed algorithm exhibits a smaller maximum delay when fewer tasks need to be processed. This outcome is attributed to the availability of more idle computing resources for each task partition.

Fig. \ref{fig10}(b) and (c) depict the convergence process of the different algorithms.  In Fig. \ref{fig10}(b), the ratio of the maximum delay achieved with different task partitioning strategies to the maximum delay achieved with a fixed partitioning action is shown. The OTPPS curve demonstrates near-optimal performance compared to the Exhaustive curve, while surpassing the K-means curve in terms of convergence and maximum delay. The proposed OTPPS algorithm reduces the maximum delay by approximately 50$\%$ compared to the MINGRA scheme. In addition, compared to other DRL algorithms, OTPPS shows better convergence performance and the ability to discover and exploit optimal actions. As analyzed in Sections \ref{II-F}, there are a large number of optional task partitioning actions. However, only a few actions are truly superior, which DQN struggles to identify and leverage. Thus, the DQN curve converges to a locally optimal value. Furthermore, a large number of action choices make it very difficult to explore the rewards of each action, resulting in a prolonged learning process. As Section \ref{II-G} describes, Actor-Critic faces convergence challenges confirmed in simulations. Fig. \ref{fig10}(c) illustrates the ratio of the maximum delay achieved with different task scheduling algorithms to the maximum delay achieved with the Local scheme. Among the various task scheduling algorithms, the OTPPS algorithm exhibits superior performance in minimizing the maximum task processing delay. The Greedy, Kuhn-Munkres, and OTPPS curves appear similar because the DRL continuously optimizes the task partitioning scheme, enabling subsequent parallel scheduling operations to achieve a smaller maximum task processing delay.
\begin{figure*}[b]
\centering
\subfloat[]{\includegraphics[width=0.33\textwidth]{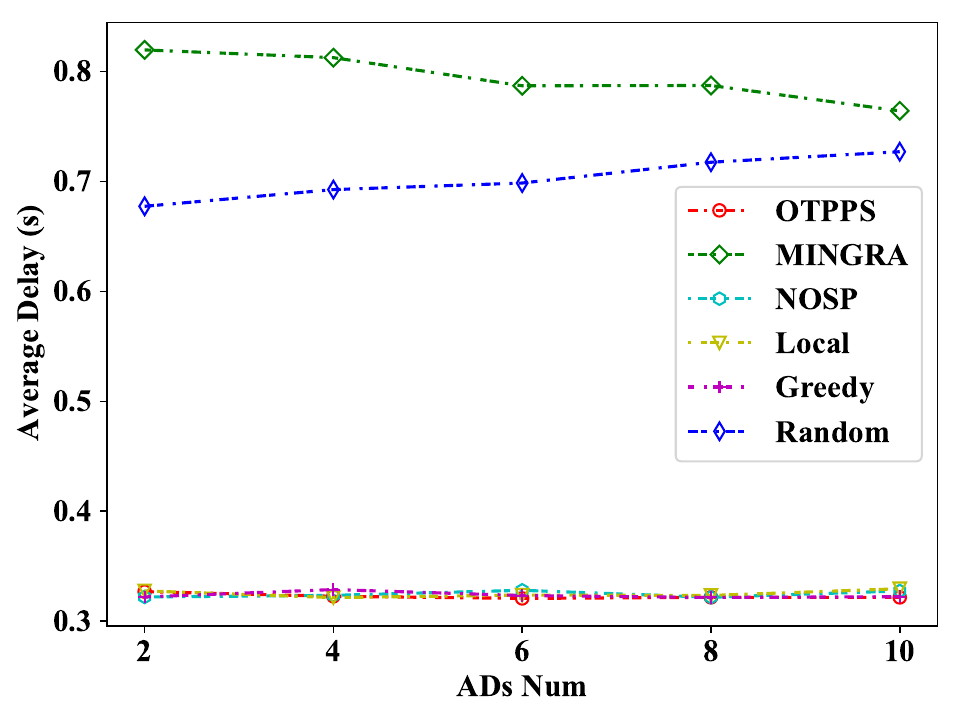}%
\label{fig_first_case}}
\subfloat[]{\includegraphics[width=0.33\textwidth]{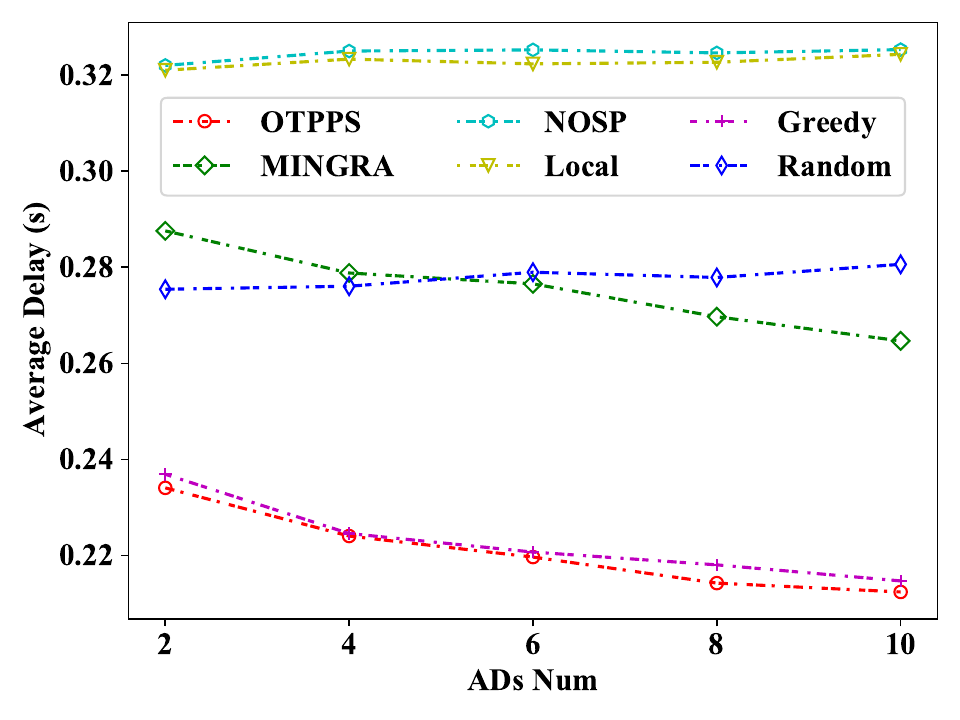}%
\label{fig_second_case}}
\subfloat[]{\includegraphics[width=0.33\textwidth]{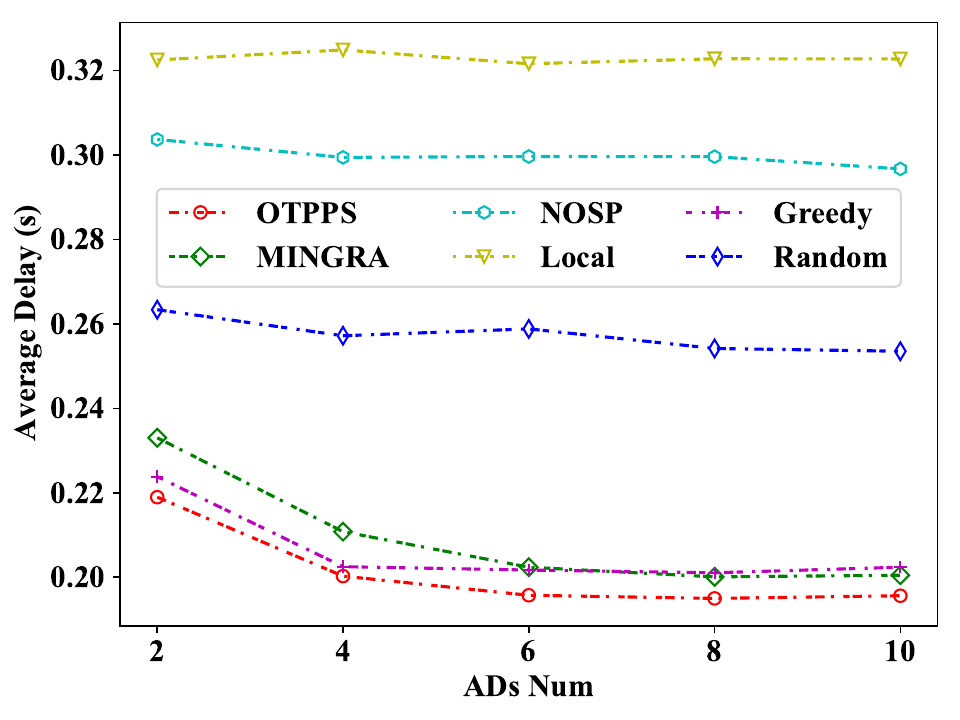}%

\label{fig_second_case}}

\caption{Average task processing delay performance of different algorithms with different numbers of auxiliary IoT devices. (a) $f_m^t\in[0.2,0.4]$GHz. (b) $f_m^t\in[0.8,1.0]$GHz. (c) $f_m^t\in[1,4,1.6]$GHz.}
\label{fig11}
\end{figure*}

\begin{table}[h]
\centering
\caption{Comparisons of CPU Execution Latency}
\label{tab4}
\resizebox{\columnwidth}{!}{%
\begin{tabular}{c|ccccc}
\hline
Number of Tasks & 2       & 3       & 4       & 5       & 6       \\ \hline
OTPPS           & 0.00497s & 0.00721s & 0.01298s & 0.01939s & 0.03495s \\
K-Means         & 0.04864s & 0.07284s & 0.09491s & 0.10652s & 0.12584s \\
MINGRA          & 0.00109s & 0.00113s & 0.00126s & 0.00148s & 0.00159s \\
Exhaustive      & 0.29105s & 2.32838s &         &         &         \\ \hline
\end{tabular}%
}
\end{table}

\emph{2) Algorithm Average Running Time Performance:} In this experiment, we evaluate the average running time of different task partitioning algorithms. It is important to note that the execution latency for various algorithms listed in Table \ref{tab4} is averaged over 10,000 independent wireless channel realizations, including policy generation and DNN training. As shown in Table \ref{tab4}, the OTPPS algorithm demonstrates a running time comparable to that of the MINGRA scheme. Specifically, it generates task partitioning actions and parallel scheduling policies in less than 0.1 seconds when $N=6$, while K-Means exhibits four times longer running time. Although the Exhaustive scheme can identify the optimal task partitioning action, its running time exceeds the range of channel coherence time when $N > 3$ \cite{49,50}, rendering it impractical in time-varying channel environments. Therefore, OTPPS enables real-time task partitioning and parallel scheduling for device-assisted MENs in fading environments.

Table \ref{tab5} presents the computational complexities corresponding to the algorithms listed in Table \ref{tab4},  where $I_{max\_iter}$ denotes the maximum number of iterations set in the K-means algorithm. The complexity analysis for OTPPS and Exhaustive algorithms can be found in Sections \ref{IV-C} and \ref{II-F}, respectively. Due to space constraints, we have not included the complexity analysis for the K-Means and MINGRA algorithms. It’s important to note that even though OTPPS and MINGRA share the same computational complexity, the runtime of OTPPS exceeds that of MINGRA. This discrepancy arises because computational complexity describes the relationship between algorithm execution time and data size growth and does not directly reflect the algorithm’s actual runtime.

\begin{figure*}[t]
\centering
\subfloat[]{\includegraphics[width=0.33\textwidth]{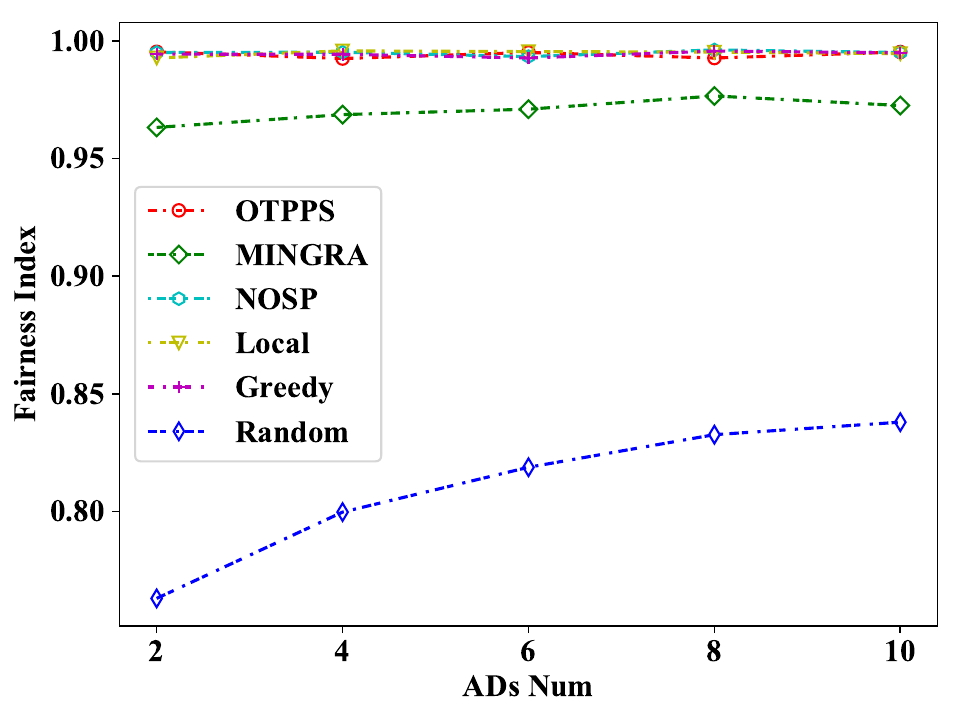}%
\label{fig_first_case}}
\subfloat[]{\includegraphics[width=0.33\textwidth]{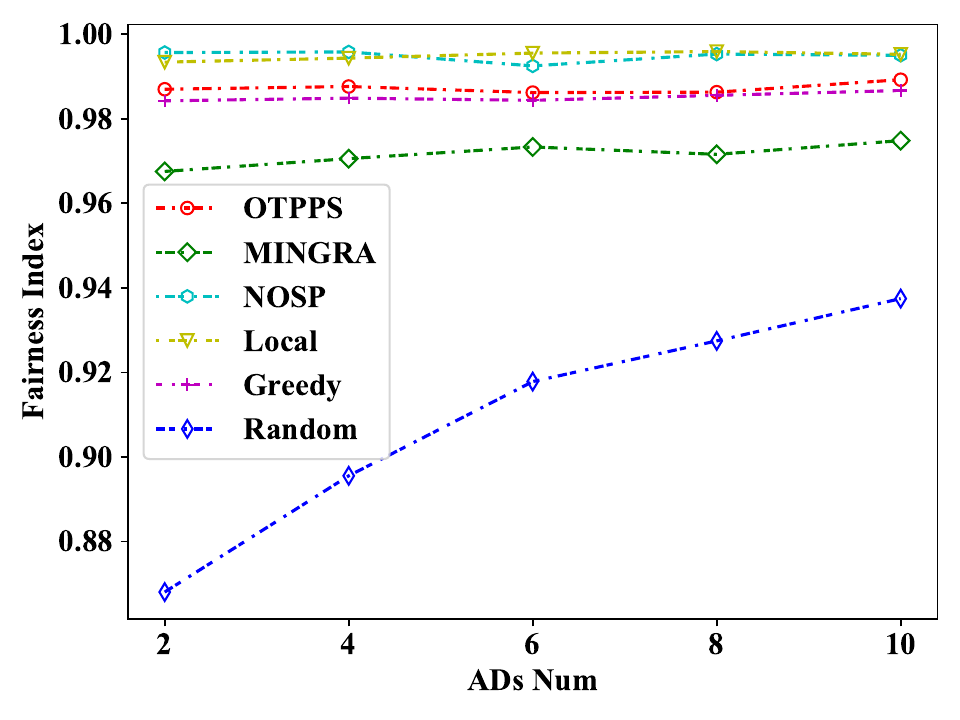}%
\label{fig_second_case}}
\subfloat[]{\includegraphics[width=0.33\textwidth]{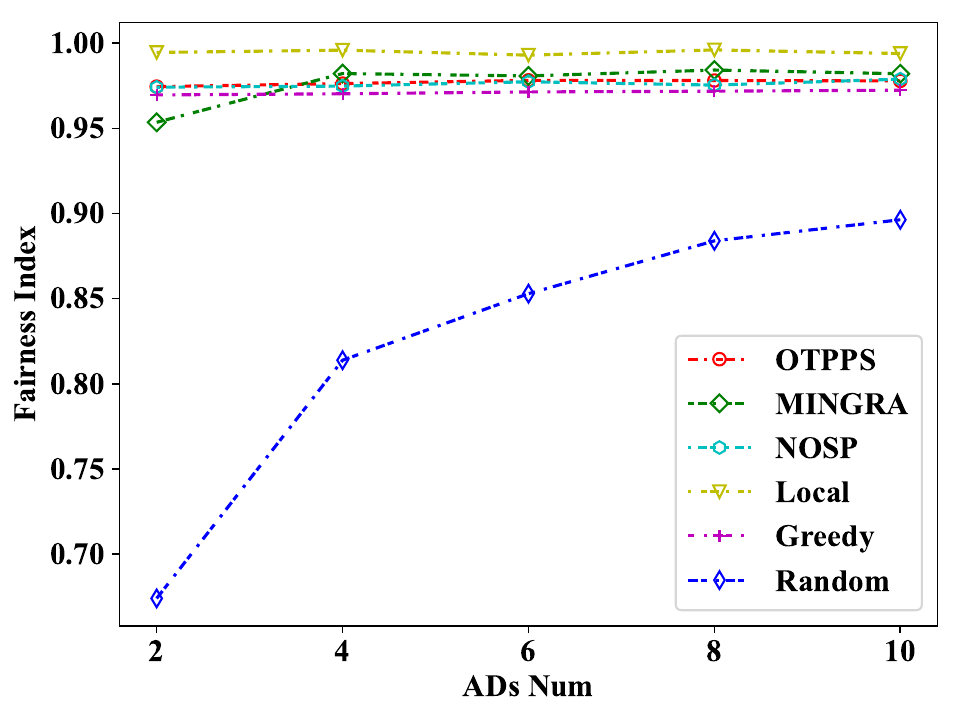}%
\label{fig_second_case}}

\caption{Fairness indices for different algorithms with different numbers of auxiliary devices. (a) $f_m^t\in[0.2,0.4]$GHz. (b) $f_m^t\in[0.8,1.0]$GHz. (c) $f_m^t\in[1,4,1.6]$GHz.}
\label{fig12}
\end{figure*}

\begin{table}[]
\centering
\caption{Computational Complexity Analysis}
\label{tab5}
\resizebox{\columnwidth}{!}{%
\begin{tabular}{@{}ll@{}}
\toprule
Algorithm\qquad\qquad\qquad\qquad  & Computational
complexity                                                                               \\ \midrule
OTPPS\qquad\qquad\qquad\qquad      & $\mathcal{O}(ZM^2N^2)$                                                                         \\
K-Means\qquad\qquad\qquad\qquad    & $\mathcal{O}(N^2I_{max\_iter}+ZN^2M^2)$                                                        \\
MINGRA\qquad\qquad\qquad\qquad     & $\mathcal{O}(ZM^2N^2)$                                                                         \\
Exhaustive\qquad\qquad\qquad\qquad & $\mathcal{O}(\prod_{n\in\mathcal{N}}\sum_{i=1}^{N_{n,res}}\left(i\right)^{N_{n,res}}+ZM^2N^2)$ \\ \bottomrule
\end{tabular}%
}
\end{table}

\emph{3) Average Task Processing Delay Performance:} This experiment primarily investigates the impact of increasing the number of auxiliary IoT devices on the average task processing delay performance of various algorithms, considering different average values of idle resources for the auxiliary IoT devices. In this simulation, $N=2$ and ES $0$ allocates 1GHz of computing resources to each TD. The idle computational resources of the auxiliary IoT devices in Fig. \ref{fig11}(a) are uniformly distributed between 0.2 and 0.4 GHz. As the number of auxiliary IoT devices increases, the average delay remains constant for the Local, NOSP, Greedy, and OTPPS schemes, while MINGRA exhibits the worst performance. This result suggests that not splitting tasks is the optimal choice when the idle computing resources of  all the auxiliary IoT devices are significantly smaller than the computing resources allocated by ES $0$. The gradual degradation in the performance of the Random scheme with an increasing number of auxiliary IoT devices can be attributed to the growing probability of scheduling task partitions for execution on the auxiliary IoT devices.

The idle computing resources of the auxiliary IoT devices in Fig. \ref{fig11}(b) are uniformly distributed between 0.8 and 1 GHz. It can be observed that the Local and NOSP curves overlap, while the OTPPS algorithm demonstrates optimal performance. This result highlights the advantages of task partitioning and indicates that partitioning tasks into smaller granularities does not necessarily yield better performance. 

The idle computing resources of the auxiliary IoT devices in Fig. \ref{fig11}(c) are uniformly distributed between 1.4 and 1.6 GHz. It can be observed that the OTPPS and MINGRA curves overlap when the number of ADs is sufficient. The NOSP scheme outperforms the Local scheme, as scheduling task partitions to ADs for execution proves more effective than performing them on ES $0$ when ADs have more computing resources than that of the ES $0$ allocation. 

The aforementioned scenarios demonstrate that OTPPS is capable of achieving optimal performance under various environmental conditions.

\emph{4) Fairness Performance:} The environment setup in this experiment is consistent with the previous experiments. From Fig. \ref{fig12}, it can be observed that the OTPPS algorithm sacrifices some fairness metrics in order to achieve lower delay, while still maintaining a high level of fairness index. OTPPS’s fairness performance is better than that of Greedy, and it improves the fairness index by about 30$\%$ compared to the Random scheme. Additionally, although OTPPS has a slightly lower fairness index compared to Local, NOSP, and MINGRA in some scenarios (Fig. \ref{fig12}(b) and Fig. \ref{fig12}(c)), it significantly reduces the average task processing delay compared to these algorithms (Fig. \ref{fig11}). The simulations demonstrate OTPPS reduces average task delay while maintaining user fairness.

\section{Conclusion}
In this paper, we propose OTPPS, a novel framework for computing offloading in device-assisted mobile edge networks. OTPPS aims to reduce the average task processing delay while ensuring fairness among users through joint optimization of task partitioning and parallel scheduling. The algorithm leverages past task partitioning experiences to enhance its partitioning action generated by a DNN through reinforcement learning. The STQ algorithm is developed to facilitate fast convergence of the optimization process. Additionally, the FDMTS algorithm is devised to address the parallel scheduling problem with a given task partitioning  action. Evaluation results demonstrate that OTPPS achieves nearly optimal average delay performance while significantly reducing the runtime by more than an order of magnitude. Furthermore, it consistently ensures high fairness indices across various system states. While the device-assisted mobile edge network alleviates the burden on the edge computing server and enhances system performance, it introduces additional communication overhead. Hence, future work will explore the incorporation of optimal communication resource allocation within the objective function.

\bibliographystyle{IEEEtran}
\bibliography{IEEEabrv, references}

\begin{IEEEbiography}[{\includegraphics [width=1in,height=1.25in] {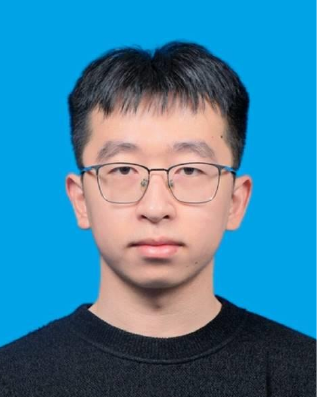}}] {Yang Li}
  received the B.S. degree in communication engineering from Beijing University of Posts and Telecommunications (BUPT), Beijing, China, in 2022. He is currently pursuing the Ph.D. degree with the Key Laboratory of Universal Wireless Communication, School of Information and Communication Engineering, BUPT. His  research interests include device-assisted mobile edge networks, computing offloading and resource allocation.
\end{IEEEbiography}

\begin{IEEEbiography}[{\includegraphics [width=1in,height=1.25in] {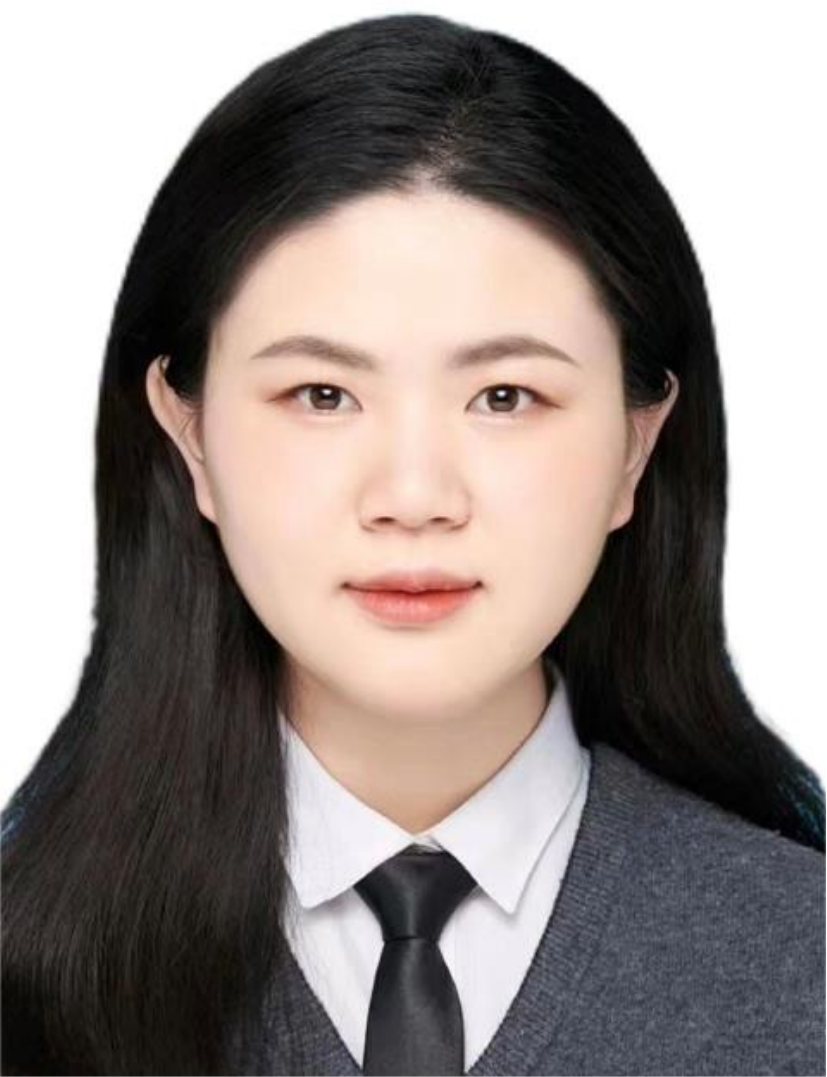}}] {Xinlei Ge}   
  received the B.S. degree in telecommunications engineering with management from Beijing University of Posts and Telecommunications (BUPT), Beijing, China, in 2023. She is currently pursuing the M.S. degree with the Key Laboratory of Universal Wireless Communication, School of Information and Communication Engineering, BUPT. Her research interests include computing power network, federated learning and service migration.
\end{IEEEbiography}

\begin{IEEEbiography}[{\includegraphics [width=1in,height=1.35in] {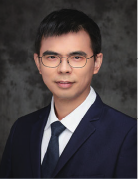}}] {Bo Lei}
  is now a director of future network research center of China telecom research institute. Bo Lei received his Master’s degree in Telecommunication Engineering from Beijing University of Posts and Telecommunications, Beijing, P. R. China, in 2006. His currently research interests include future network architecture, new network technology, computing power network and 5G application verification. Bo Lei now leads the future network research center focusing on future network. He is the first author of two technical books and has published more than 30 papers in top journals and international conferences, and filed more than 30 patents.
\end{IEEEbiography}

\begin{IEEEbiography}[{\includegraphics [width=1in,height=1.25in] {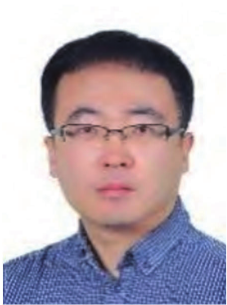}}] {Xing Zhang} 
  (M’10-SM’14) is currently a full professor with the School of Information and Communications Engineering, Beijing University of Posts and Telecommunications, China. His research interests are mainly in 5G/6G networks, satellite communications, edge intelligence, and Internet of Things. He has authored or coauthored five technical books and over 300 papers in top journals and international conferences and holds over 80 patents. He has received six Best Paper Awards in international conferences. He is a Senior Member of the IEEE and a member of CCF. He has served as a General Co-Chair of the third IEEE International Conference on Smart Data (SmartData-2017), as a TPC Co-Chair/TPC Member for a number of major international conferences.
\end{IEEEbiography}

\begin{IEEEbiography}[{\includegraphics [width=1in,height=1.35in] {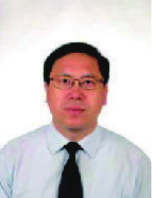}}] {Wenbo Wang}
  received the B.S., M.S., and Ph.D. degrees from BUPT in 1986, 1989, and 1992, respectively. He is currently a Professor with the School of Information and Communications Engineering, and the Executive Vice Dean of the Graduate School, Beijing University of Posts and Telecommunications. He is currently the Assistant Director with the Key Laboratory of Universal Wireless Communication, Ministry of Education. He has authored over 200 journal and international conference papers, and six books. His current research interests include radio transmission technology, wireless network theory, and software radio technology.
\end{IEEEbiography}

\end{document}